\documentclass[11pt]{article}             
\usepackage{osid}                         %
\usepackage[T1]{fontenc}
\usepackage[utf8]{inputenc}
\usepackage{amsmath,amssymb,mathtools,graphicx,amsthm}
\usepackage{bm,soul}
\usepackage{hyperref}
\usepackage{enumitem}
\usepackage{dsfont,bm}
\usepackage{float}
\usepackage[scr=boondoxo,scrscaled=1.05]{mathalfa}
\usepackage{centernot}
\usepackage{caption}
\usepackage{aurical}
\hypersetup{
	colorlinks   = true, 
	urlcolor     = blue, 
	linkcolor    = blue, 
	citecolor   = red 
}
\usepackage{tikz}
\usetikzlibrary{calc,decorations.markings,decorations.pathmorphing,decorations.text,backgrounds}
\tikzset{snake it/.style={decorate, decoration=snake}}
\usepackage{pgfplots} 
\usepackage{epigraph} 
\usepackage{fancyhdr}
\newcommand{\ti}{t_{\mathfrak{0}}}
\newcommand{\tf}{t_{\mathit{f}}}
\title{On the unraveling of open quantum dynamics.}
\author{Brecht I. C Donvil \\{\footnotesize\it Institute for Complex Quantum Systems and IQST, Ulm University - Albert-Einstein-Allee 11, D-89069 Ulm, Germany \& brecht.donvil@uni-ulm.de}\\[2ex]
        Paolo Muratore-Ginanneschi 
                     \\{\footnotesize\it University of Helsinki, Department of Mathematics and Statistics
                     	P.O. Box 68 FIN-00014, Helsinki, Finland \& paolo.muratore-ginanneschi@helsinki.fi} }

\begin{document}

\maketitle
\begin{abstract}
 It is well known that the state operator of an open quantum system can be generically represented as the solution of a time-local equation -- a quantum master equation.
 Unraveling in quantum trajectories offers a picture of open system dynamics dual to solving master equations. In the unraveling picture, physical indicators are computed as Monte-Carlo averages over a stochastic process valued in the Hilbert space of the system. This approach is particularly adapted to simulate systems in large Hilbert spaces. 
 We show that the dynamics of an open quantum system generically admits an unraveling in the Hilbert space of the system described by a Markov process generated by ordinary stochastic differential equations for which rigorous concentration estimates are available. The unraveling can be equivalently formulated in terms of norm-preserving state vectors or in terms of linear \textquotedblleft ostensible\textquotedblright\ processes trace preserving only on average. We illustrate the results in the case of a two level system in a simple boson environment. 
 Next, we derive the state-of-the-art form of the Di\'osi-Gisin-Strunz Gaussian random ostensible state equation in the context of a model problem. This equation provides an exact unraveling of open systems in Gaussian environments. We compare and contrast the two unravelings and their potential for applications to quantum error mitigation.  
\end{abstract}

\begin{minipage}{0.5\textwidth}
		\begin{minipage}{0.8\textwidth}
	\epigraph{ ``I have an equation; do
		you have one too?''}{\textit{Dirac's first question to Feynman \\ according to \cite{ZeeA2010} pag. 105}}
\end{minipage}
\end{minipage}
\begin{minipage}{0.5\textwidth}
	\begin{minipage}{0.8\textwidth}
		\begin{center}
\textit{To the memory of \\  G\"oran Lindblad \\who has an equation.}
\end{center}
\end{minipage}
\end{minipage}

\section{Introduction}

A classical result in the theory of open quantum systems states that the memory kernel equation obtained by means of the Nakajima-Zwanzig projection operator method \cite{NakS1958,ZwaR1960} is always amenable to the form of a master equation on a finite time interval \cite{FuKr1968,ShiT1970,VstV1973,GrTaHa1977,vaWoLe1995,vaWoLe2000a,AnCrHa2007,ChKo2010}. In finite dimensional Hilbert spaces, the master equation also admits an unique canonical form \cite{HaCrLiAn2014} which makes the analysis of the solution semi-group particularly transparent.  The necessary and sufficient conditions for a completely positive semi-group -- whose derivation is the eternal glory of  Lindblad \cite{LinG1976} and Gorini, Kossakowski, and Sudarshan \cite{GoKoSu1976}--  translate into the positivity of the scalar quantities that couple the decoherence operators to the canonical form of the master equation. For this reason the couplings are often referred to as canonical ``rates''. The derivation of the general master equation, however, does not impose any sign constraint on the canonical rates. It is an open problem to determine 
the conditions guaranteeing that particular solutions enjoy complete positivity \cite{HaCrLiAn2014,ChrD2022}. 

One avenue to analyze and numerically solve master equations is quantum trajectory theory formulated in terms of the stochastic Schr\"odinger equation \cite{BaBe1991,CarH1993} or of quantum state diffusion \cite{GisN1984,GaPaZo1992}. The avenue becomes especially useful as the dimension of the system Hilbert space increases \cite{DaCaMo1992}.  The gist of the method is to compute the solution of the master equation --unravel the master equation \cite{CarH1993}-- as a Monte-Carlo average over solutions of ordinary stochastic differential equations in the Hilbert space of the system. So far, quantum trajectories enjoying these properties appeared only  as a ``Lagrangian picture'' in hydrodynamic sense (see e.g. \S~11.2 \cite{CeCeVu2009}) of the Lindblad-Gorini-Kossakowski-Sudarshan completely positive master equation. In this contribution, we explain how the same stochastic differential equation as in the completely positive case yields the unraveling of the most general canonical form of the master equation at the same computational cost \cite{DoMG2022,DoMG2023}. The main point is to respect at the level of Monte-Carlo average the canonical algebraic form of the solution semi-group. Namely, in the general case the semi-group is only completely bounded \cite{PauV2003} and as such is amenable to the difference of two non-trace preserving completely positive maps at any instant of time \cite{WitG1981,PauV1982}. This task is accomplished by the introduction of a martingale process which we call the \textquotedblleft influence martingale\textquotedblright  \cite{DoMG2022}. The choice of the name emphasizes that the result holds irrespective of the environment, including of course a Gaussian one \cite{FeVe1963}.
Gaussian, bosonic or fermionic, environments linearly coupled to the system are, among other reasons, important because it is in that case possible to derive an exact pair of random equations governing the evolution of vectors unraveling the system state operator. The vectors obey a linear evolution equation and {preserve the norm squared only on average}. For this reason, following \cite{WisH1996} we refer to {them} as \textquotedblleft ostensible\textquotedblright\ state vectors. The result is a refinement \cite{StGr2002,StoJ2004,TilA2017} see also \cite{LiSt2017} of an approach pioneered by Strunz \cite{StrW1996}, in collaboration with Di\'osi and Gisin \cite{DiSt1997,DiGiSt1998}. {In the second part of the paper we derive the state-of-the-art  form of the Di\'osi-Gisin-Strunz random ostensible state vector equation in a model problem, yet capturing all the features of the general case}. 
{By presenting} in parallel the derivation of these two unravelings, we hope to draw the interest of the reader to the theory of unraveling in quantum trajectories and to its significance for parameter estimation \cite{WeMuKiScRoSi2016} and error mitigation \cite{TeBrGa2016}.   This paper is, however, not meant to be a comprehensive review of the theory.

The structure of the paper is as follows. In section~\ref{sec:me} we recall the properties of the canonical master equation. In section~\ref{sec:LGKS} we recall how to unravel the canonical master equation under the hypotheses guaranteeing the complete positivity of the semi-group. These two sections report well known material. The reason to briefly reproduce it here is to highlight the logic leading to the most general form of the unraveling. Readers familiar with the topics can directly proceed to section~\ref{sec:cb} where we introduce the influence martingale. In section~\ref{sec:ostensible} we discuss the generation of ostensible statistics, a numerical application of quantum trajectory theory particularly suited for quantum parameter estimation \cite{GaWi2001}.
In section~\ref{sec:tls} drawing from \cite{SmVa2010} we discuss an application to the spin boson model. For simplicity of discussion we restrict the attention to the case when explicit formulas are exact. The reference to \cite{SmVa2010} makes transparent how the general case can be handled using time-convolutionless perturbation theory \cite{NeBrWe2021}. Section~\ref{sec:DGS} is devoted to {a} model problem illustration of the Di\'osi-Gisin-Strunz formalism \cite{DiGiSt1998}. In section~\ref{sec:comparison} we contrast the two unravelings.  
The last section is devoted to conclusions and outlook.

\section{The canonical master equation of open quantum dynamic}
\label{sec:me}

Under generic conditions
 \cite{FuKr1968,ShiT1970,VstV1973,GrTaHa1977,vaWoLe1995,vaWoLe2000a,AnCrHa2007,ChKo2010}, the state operator of an open system on a $ d $-dimensional Hilbert space $ \mathcal{H} $ obeys a master equation amenable to the following unique canonical form
	\begin{subequations}
	\label{me:me}
	\begin{align}
		&\label{me:me1}
		\partial_{t}\bm{\rho}_{t}=-\imath\,\left[\operatorname{H}_{t}\,,\bm{\rho}_{t}\right]
		+\sum_{\mathscr{l}=1}^{d^{2}-1}\mathscr{w}_{\mathscr{l}; t}\operatorname{D}_{\operatorname{L}_{\mathscr{l};t}}(\bm{\rho}_{t})
		\\
		&
		\operatorname{D}_{\operatorname{L}_{\mathscr{l}; t}}(\bm{\rho}_{t})=
		\frac{1}{2}
		\left(
		\left[\operatorname{L}_{\mathscr{l}; t}\,,\bm{\rho}_{t}\operatorname{L}_{\mathscr{l}; t}^{\dagger}\right]
		+
		\left[\operatorname{L}_{\mathscr{l}; t}\bm{\rho}_{t}\,,\operatorname{L}_{\mathscr{l}; t}^{\dagger}\right]
		\right)
		\label{me:me2}
	\end{align}
\end{subequations}
The properties of the canonical form are thoroughly investigated in \cite{HaCrLiAn2014}. Here we only recall few facts that are essential for  its  unraveling (\ref{me:me}).
	\begin{enumerate}[nosep, leftmargin=10pt, label=\roman*]
	\item \label{me:i}The canonical decoherence couplings $ \left\{  \mathscr{w}_{\mathscr{l}; t}\right\}_{\mathscr{l}=1}^{d^{2}-1} $ are in general real, scalar functions of the time variable $ t $. Physically, {one may conceptualize $\mathscr{w}_{\mathscr{l}; t}$ as the time derivative} of the probability that the open systems interacts with its surrounding environment via the $ \mathscr{l} $-th decoherence operator $ \operatorname{L}_{\mathscr{l}; t} $. As (\ref{me:me}) holds without any reference to the van-Hove scaling limit, there is no reason to expect the canonical couplings to be sign definite\footnote{In \cite{HaCrLiAn2014,DoMG2023} they are referred as  \textquotedblleft canonical rates\textquotedblright\ and in \cite{DoMG2022} as \textquotedblleft weights\textquotedblright. The terminology  is in both cases somewhat misleading because it is commonly applied to quantities positive by definition}.
	\item \label{me:ii} At any instant of time $ t $, the collection $ \left\{  \operatorname{L}_{\mathscr{l}; t}\right\}_{\mathscr{l}=1}^{d^{2}-1} $  of decoherence operators together  with a properly normalized identity matrix on $ \mathcal{H} $
	\begin{align}
		\operatorname{L}_{0}=\frac{\operatorname{1}_{d}}{\sqrt{d}}
		\nonumber
	\end{align} 
	forms  a complete basis of the space $ \mathcal{M}_{d} $ of complex matrices on $ \mathcal{H} $ with respect to the Hilbert-Schmidt inner product
	\begin{align}
		\operatorname{Tr}\left(\operatorname{L}_{\mathscr{l}; t}^{\dagger} \operatorname{L}_{\mathscr{k}; t}\right) =\delta_{\mathscr{l},\mathscr{k}}\hspace{1.0cm}\mathscr{l},\mathscr{k}=0,\dots, d^{2}-1
				\label{me:ortho}
	\end{align}
	\item \label{me:iii} $ \operatorname{H}_{t} $ is a self-adjoint element of $ \mathcal{M}_{d} $, eventually also time dependent.
\end{enumerate}
The orthonormality relations (\ref{me:ortho}) are always emphasized in the literature while discussing the derivation of the Lindblad-Gorini-Kossakowski-Sudarshan
master equation. Less well known is the following consequence of the fact that $\left\{  \operatorname{L}_{\mathscr{l}; t}\right\}_{\mathscr{l}=1}^{d^{2}-1} $ spans the sub-space of trace-less matrices:
\begin{theorem}{Proposition}
	\begin{subequations}
				\label{me:povm}
			\begin{align}
&		\sum_{\mathscr{l}=1}^{\mathscr{L}}\operatorname{L}_{\mathscr{l}; t}^{\dagger}\operatorname{L}_{\mathscr{l}; t}=\mathscr{g}\operatorname{1}_{d}
		\label{cb:povm1}
		\\
&		\mathscr{g}=\frac{d^{2}-1}{d}
\label{me:povm2}
	\end{align}
\end{subequations}	
\end{theorem}
\begin{proof}~\newline
To prove the claim, we notice that each of the $ \operatorname{L}_{\mathscr{l}; t} $'s must be connected to a complete set of generators 
$ \operatorname{T}_{\mathscr{l}} $ of the $ \mathrm{SU}(d) $ Lie algebra by a unitary transformation
\begin{align}
&	\operatorname{L}_{\mathscr{l}; t}=\sum_{\mathscr{k}=1}^{d^{2}-1}\operatorname{T}_{\mathscr{k}} U_{\mathscr{k},\mathscr{l}; t}
\nonumber\\
&U_{\mathscr{k},\mathscr{l}; t}:=\operatorname{Tr}(\operatorname{T}_{\mathscr{k}}\operatorname{L}_{\mathscr{l}; t})
	\nonumber
\end{align}
The identity holds because each of the two collections of matrices forms a complete orthonormal basis of the sub-space of trace-less elements of $ \mathcal{M}_{d} $. Unitarity implies that
\begin{align}
	\sum_{\mathscr{l}=1}^{d^{2}-1}\operatorname{L}_{\mathscr{l}; t}^{\dagger}\operatorname{L}_{\mathscr{l}; t}={\sum_{\mathscr{l}=1}^{d^{2}-1}}
	\sum_{\mathscr{i},\mathscr{j}=1}^{d^{2}-1}\operatorname{T}_{\mathscr{i}} \operatorname{T}_{\mathscr{j}}\bar{U}_{\mathscr{i},\mathscr{l}; t}U_{\mathscr{j},\mathscr{l}; t}=
	\sum_{\mathscr{i}=1}^{d^{2}-1}\operatorname{T}_{\mathscr{i}}^{2} 
	\nonumber
\end{align}
Finally,  we can invoke Schur's lemma or, equivalently, use known properties of the $ \mathrm{SU}(d) $ Lie algebra (see e.g. formula (16) of \cite{HabH2021}\footnote{note, however, that our choice of normalization  of the $ \operatorname{T}_{\mathscr{i}} $'s differs from \cite{HabH2021} by a factor $ \sqrt{2} $}) 
\begin{align}
&	\operatorname{T}_{\mathscr{i}}^{2} =\frac{\operatorname{1}_{d}}{d}+\sum_{\mathscr{j}=1}^{d^{2}-1}\mathscr{d}_{\mathscr{i},\mathscr{i},\mathscr{j}}\operatorname{T}_{\mathscr{j}}
\nonumber\\
& \sum_{\mathscr{i}=1}^{d^{2}-1}\mathscr{d}_{\mathscr{i},\mathscr{i},\mathscr{j}}=0 & \forall\,j
	\nonumber
\end{align}
to prove (\ref{me:povm}).
\end{proof}
In consequence of \ref{me:i}-\ref{me:iii}, the canonical master equation is readily trace and self-adjoint preserving. Under reasonable regularity of assumptions on the dependence of canonical couplings, Hamiltonian, and decoherence operators,
 solutions of (\ref{me:me}) exist and are unique, globally. 
This result is better appreciated if we lift (\ref{me:me}) to a linear system of differential equations by means of the channel-state isomorphism see e.g.
\cite{ChrD2022}  or \S~4.2 of \cite{HoJo1991}. To this goal, we introduce the lexicographic map (see e. g. chapter~9 of \cite{BeZy2006})
\begin{align}
	\ell \colon \mathbb{N}\,\times\,\mathbb{N}\mapsto \mathbb{N}
	\nonumber
\end{align}
defined as
\begin{align}
	\ell (i,j)=d\,(i-1)+j
	\nonumber
\end{align}
and we use it to establish a one-to-one correspondence between elements of $ \mathcal{M}_{d} $ and vectors in $ \mathbb{C}^{d^{2}} $ by 
{reshaping} \cite{MisJ2011}
\begin{align}
	\left \langle\,\operatorname{e}_{\ell(i,j)}^{(d^{2})}\,,{\operatorname{res}}(\bm{\rho}_{t})\,\right\rangle=
	\left \langle\,\operatorname{e}_{i}^{(d)}\,,\bm{\rho}_{t}  \operatorname{e}_{j}^{(d)}\,\right\rangle
	\nonumber
\end{align}
where $ \big{\{ }\operatorname{e}_{i}^{(d)}  \big{\}}_{i=1}^{d}$ and $ \big{\{} \operatorname{e}_{i}^{(d^{2})} \big{\}}_{i=1}^{d^{2}} $ are respectively the canonical bases of the complex Euclidean spaces $ \mathbb{C}^{d} $ and $ \mathbb{C}^{d} $ while $ \left \langle\,\cdot,\,\cdot\right\rangle $ denotes the natural inner product. Once equipped with these definitions we see that (\ref{me:me}) becomes 
\begin{align}
	\frac{\mathrm{d}}{\mathrm{d} t}{\operatorname{res}}(\bm{\rho}_{t})=\mathscr{L}_{t}{\operatorname{res}}(\bm{\rho}_{t})
	\nonumber
\end{align}
where $ \mathscr{L} $ is the $ d^{2}\,\times\,d^{2} $ matrix
\begin{align}
 &\mathscr{L}_{t} =-\imath\,\big{(}\operatorname{H}_{t}\,\otimes\,\operatorname{1}_{d}-\operatorname{1}_{d}\,\otimes\,\operatorname{H}_{t}^{\top}\big{)}
 +\sum_{\mathscr{l}=1}^{d^{2}-1}\mathscr{w}_{\mathscr{l}; t}\mathscr{D}_{\operatorname{L}_{\mathscr{l};t}}
 \nonumber\\
 &
 \mathscr{D}_{\operatorname{L}_{\mathscr{l}; t}}=
 \operatorname{L}_{\mathscr{l}; t}\,\otimes\,\bar{\operatorname{L}}_{\mathscr{l}; t}
 -\frac{\operatorname{L}_{\mathscr{l}; t}^{\dagger}\operatorname{L}_{\mathscr{l}; t}\,\otimes\,\operatorname{1}_{d}+\operatorname{1}_{d}\,\otimes\,\operatorname{L}_{\mathscr{l}; t}^{\top}\bar{\operatorname{L}}_{\mathscr{l}; t}}{2}
	\nonumber
\end{align}
It becomes then evident how to apply the standard theorems of existence and uniqueness of linear ordinary differential equations \cite{CeCeVu2009}.
Furthermore, we can always write the solution of (\ref{me:me}) as
\begin{align}
	{\operatorname{res}}(\bm{\rho}_{t})=\mathsf{F}_{t,s}{\operatorname{res}}(\bm{\rho}_{s})
	\nonumber
\end{align}
where $ \mathsf{F}_{t,s}\in \mathcal{M}_{d^{2}} $ is a linear flow solution of
\begin{align}
&	\frac{\mathrm{d}}{\mathrm{d} t}\mathsf{F}_{t,s}=\mathscr{L}_{t}\mathsf{F}_{t,s}
	\nonumber\\
&\mathsf{F}_{t,t}=\operatorname{1}_{d^{2}}& \forall\,t
\nonumber	
\end{align}
An important consequence of {reshaping} stems from a paramount observation made by Sudarshan, Mathews and Rau in \cite{SuMaRa1961}.
Self-adjoint preservation translates into the condition
\begin{align}
	\mathsf{F}_{t,s}^{(\mathsf{R})\dagger}=\mathsf{F}_{t,s}^{(\mathsf{R})}
	\nonumber
\end{align}
where $ \mathsf{F}_{t,s}^{(\mathsf{R})} $ is the reshuffled (or Choi) matrix whose components are related to that of the flow
by the relations
\begin{align}
\left \langle\,\bm{e}_{\ell(i,j)}^{(d^{2})}	\,,\mathsf{F}_{t,s}^{(\mathsf{R})}\bm{e}_{\ell(m,n)}^{(d^{2})}\,\right\rangle=	\left \langle\,\bm{e}_{\ell(i,m)}^{(d^{2})}	\,,\mathsf{F}_{t,s}\bm{e}_{\ell(j,n)}^{(d^{2})}\,\right\rangle
\label{me:reshuffled}
\end{align}
The reshuffled matrix is also adapted to analyze positivity of the dynamics: being self-adjoint, it has real eigenvalues $ \left\{ f_{n} \right\}_{n=1}^{d^{2}} $ and admits the spectral decomposition
\begin{align}
	\mathsf{F}_{t,s}^{(\mathsf{R})}=\sum_{n=1}^{d^{2}} \operatorname{sign}(f_{n} ) \bm{v}_{n}\bm{v}_{n}^{\dagger}
	\nonumber
\end{align}
having reabsorbed the absolute value of the $ f_{n} $'s in the normalization of the corresponding eigenvectors $ \bm{v}_{n} $. We thus arrive at
\begin{align}
	\left \langle\,\bm{e}_{\ell(i,m)}^{(d^{2})}	\,,\mathsf{F}_{t,s}\bm{e}_{\ell(j,n)}^{(d^{2})}\,\right\rangle
	=\sum_{n=1}^{d^{2}} \operatorname{sign}(f_{n; t} )
	\left \langle\,\bm{e}_{\ell(i,j)}^{(d^{2})}	\,, \bm{v}_{n; t}\,\right\rangle 
	\left \langle\,\bm{v}_{n; t}\,, \bm{e}_{\ell(m,n)}^{(d^{2})}\,\right\rangle
	\nonumber
\end{align}
which projected back to $ \mathcal{M}_{d} $ {yields} the operator sum representation of the (semi)-group solution of (\ref{me:me})
\begin{align}
	\bm{\rho}_{t}=\sum_{n=1}^{d^{2}} \operatorname{sign}(f_{n; t} )\operatorname{V}_{n; t,s}\bm{\rho}_{s}\operatorname{V}_{n; t,s}^{\dagger}
	\label{me:cb}
\end{align}
Trace preservation corresponds to the additional condition
\begin{align}
	\sum_{n=1}^{d^{2}} \operatorname{sign}(f_{n; t} )\operatorname{V}_{n; t,s}^{\dagger}\operatorname{V}_{n; t,s}=\operatorname{1}_{d}
	\nonumber
\end{align}
Choi theorem \cite{ChoM1975} shows that complete positivity in the finite dimensional case reduces to the condition that
the reshuffling of flow matrix $ \mathsf{F}_{t,s}^{(R)} $ has a positive spectrum. It becomes then manifest that complete positivity is 
a sufficient condition to map arbitrary positive operators into positive ones.
We also learn that in general solutions of (\ref{me:me}) can be written as the difference of two completely positive operator sums individually non-trace preserving.
The derivation of (\ref{me:cb}) takes advantage of the hypothesis of a finite dimensional operator space. In fact, a result in abstract linear operator algebra due to Wittstock \cite{WitG1981}  and Paulsen \cite{PauV1982} allows us to identify (\ref{me:cb}) with the canonical decomposition of a \textquotedblleft completely bounded\textquotedblright\  operator \cite{JoKrPa2009}. For our purposes, it is sufficient to observe that an element of $ \mathcal{M}_{d} $ is completely bounded if it is bounded \cite{JoKrPa2009}. We refer to  this latter reference together with \cite{ChrD2022} and chapters~10-11 of \cite{BeZy2006} for further details.

\section{Unraveling of the Lindblad-Gorini-Kossakowski-Sudarshan (completely positive) master equation}
\label{sec:LGKS}

Lindblad \cite{LinG1976} and Gorini, Kossakowski and Sudarshan \cite{GoKoSu1976} proved that the flow matrix $ \mathsf{F}_{t,s} $ with $t\,\geq\,s\,\geq\,0  $ is completely positive if and only if  the canonical couplings are positive. In such a case, {we can identify the canonical couplings} with the rates 
of a counting process
\begin{align}
	\mathscr{w}_{\mathscr{l}; t}=\mathscr{r}_{\mathscr{l}; t}\,\geq\,0 \hspace{1.0cm}& \forall\,t\in \mathbb{R}_{+}\hspace{0.5cm}\&\hspace{0.5cm}\mathscr{l}=1,\dots,d^{2}-1
	\label{LGKS:rates}
\end{align}
When the conditions (\ref{LGKS:rates}) hold true, we write the master equation as
\begin{align}
	\label{LGKS:eq}
			\partial_{t}\bm{\rho}_{t}=-\imath\,\left[\operatorname{H}_{t}\,,\bm{\rho}_{t}\right]
			+\sum_{\mathscr{l}=1}^{d^{2}-1}\mathscr{r}_{\mathscr{l}; t}\operatorname{D}_{\operatorname{L}_{\mathscr{l};t}}(\bm{\rho}_{t})
		\end{align}
and we are in the position to prove \cite{BaBe1991,CarH1993,DaCaMo1992}
\begin{theorem}{proposition}
	\label{prop:LGKS}
	The solution of (\ref{LGKS:eq}) admits the Monte-Carlo representation
	\begin{align}
		\bm{\rho}_{t}=\operatorname{E}(\bm{\psi}_{t}\bm{\psi}_{t}^{\dagger})
		\label{LGKS:main}
	\end{align}
where $ \operatorname{E} $ is the expectation over the solution paths of the It\^o \cite{KleF2005} stochastic differential equations
	\begin{subequations}
	\label{LGKS:sse}
	\begin{align}
		&\label{LGKS:sse1}
		\mathrm{d}\bm{\psi}_{t}=\mathrm{d}t\,\bm{f}_{t}+\sum_{\mathscr{l}=1}^{d^{2}-1}\mathrm{d}{\nu}_{\mathscr{l};t}\left(\frac{\operatorname{L}_{\mathscr{l}; t}\bm{\psi}_{t}}{\left\|\operatorname{L}_{\mathscr{l}; t}\bm{\psi}_{t}\right\|}-\bm{\psi}_{t}\right)
		\\
		&\label{LGKS:sse2}
		\bm{f}_{t}=-\imath\,\operatorname{H}_{t}\bm{\psi}_{t}-
		\sum_{\mathscr{l}=1}^{d^{2}-1}\mathscr{r}_{\mathscr{l};t}\frac{\operatorname{L}_{\mathscr{l}; t}^{\dagger}\operatorname{L}_{\mathscr{l}; t}-\left\|\operatorname{L}_{\mathscr{l}; t}\bm{\psi}_{t}\right\|^{2}\operatorname{1}_{d}}{2}\bm{\psi}_{t}
		\\
		&\label{LGKS:sse3}
		\bm{\psi}_{\ti}=\bm{z}
	\end{align}
\end{subequations}
and of those (obviously related) for adjoint dual $ \bm{\psi}_{t}^{\dagger} $. The It\^o stochastic differential equations (\ref{LGKS:sse}) 
are driven by $ d^{2}-1 $ counting processes 
	\begin{subequations}
	\label{LGKS:counting}
	\begin{align}
		&\label{LGKS:counting1}
		{\mathrm{d}\nu_{\mathscr{l};t}\mathrm{d}\nu_{\mathscr{k},t}=\delta_{\mathscr{l},\mathscr{k}}\mathrm{d}\nu_{\mathscr{l};t}}
		\\
		&\label{LGKS:counting2}
		\operatorname{E}\Big{(}\mathrm{d}\nu_{\mathscr{l};t}\big{|}\big{\{} \bm{\psi}_{t}\,,\bm{\psi}_{t}^{\dagger} \big{\}}\cap\mathfrak{F}_{t}\Big{)}
		=\mathscr{r}_{\mathscr{l};t}\left\|\operatorname{L}_{\mathscr{l}; t}\bm{\psi}_{t}\right\|^{2}\mathrm{d}t
	\end{align}
adapted to the natural filtration \cite{KleF2005} $ \left\{ \mathfrak{F}_{t} \right\}_{t\,\geq\,\ti} $  of 
the process $ \big{\{} \bm{\psi}_{t},\bm{\psi}_{t}^{\dagger}\big{\}}_{t\,\geq\,\ti} $.
\end{subequations}
\end{theorem}
\begin{proof}~\newline
	The proof is an immediate consequence of It\^o lemma applied to counting processes. Namely, when we insert (\ref{LGKS:sse1}) and its
	adjoint dual in 
	\begin{align}
		\mathrm{d}(\bm{\psi}_{t}\bm{\psi}_{t}^{\dagger})=
		(\mathrm{d}\bm{\psi}_{t})\bm{\psi}_{t}^{\dagger}+\bm{\psi}_{t}\mathrm{d}\bm{\psi}_{t}^{\dagger}+(\mathrm{d}\bm{\psi}_{t})\mathrm{d}\bm{\psi}_{t}^{\dagger}
		\nonumber
	\end{align}
straightforward algebra using (\ref{LGKS:counting1}) yields
\begin{align}
	\mathrm{d}(\bm{\psi}_{t}\bm{\psi}_{t}^{\dagger})= (\bm{f}_{t}\bm{\psi}_{t}^{\dagger}+\bm{\psi}_{t}\bm{f}_{t}^{\dagger})\mathrm{d}t +\sum_{\mathscr{l}=1}^{d^{2}-1}\mathrm{d}{\nu}_{\mathscr{l};t}\left(\frac{\operatorname{L}_{\mathscr{l}; t}\bm{\psi}_{t}\bm{\psi}_{t}^{\dagger}\operatorname{L}_{\mathscr{l}; t}}{\left\|\operatorname{L}_{\mathscr{l}; t}\bm{\psi}_{t}\right\|^{2}}-\bm{\psi}_{t}\bm{\psi}_{t}^{\dagger}\right)
	\nonumber
\end{align}
Next, we can invoke the telescopic property of conditional expectations to evaluate the expectation value of the stochastic differential
\begin{align}
&	\operatorname{E}\left(\mathrm{d}{\nu}_{\mathscr{l};t}\left(\frac{\operatorname{L}_{\mathscr{l}; t}\bm{\psi}_{t}\bm{\psi}_{t}^{\dagger}\operatorname{L}_{\mathscr{l}; t}}{\left\|\operatorname{L}_{\mathscr{l}; t}\bm{\psi}_{t}\right\|^{2}}-\bm{\psi}_{t}\bm{\psi}_{t}^{\dagger}\right)\right)
\nonumber\\
&	=\operatorname{E}\left(\operatorname{E}\Big{(}\mathrm{d}\nu_{\mathscr{l};t}\big{|}\big{\{} \bm{\psi}_{t}\,,\bm{\psi}_{t}^{\dagger} \big{\}}\cap\mathfrak{F}_{t}\Big{)}\left(\frac{\operatorname{L}_{\mathscr{l}; t}\bm{\psi}_{t}\bm{\psi}_{t}^{\dagger}\operatorname{L}_{\mathscr{l}; t}}{\left\|\operatorname{L}_{\mathscr{l}; t}\bm{\psi}_{t}\right\|^{2}}-\bm{\psi}_{t}\bm{\psi}_{t}^{\dagger}\right)\right)
\nonumber\\
&=\mathscr{r}_{\mathscr{l}; t}\operatorname{E}\left(\operatorname{L}_{\mathscr{l}; t}\bm{\psi}_{t}\bm{\psi}_{t}^{\dagger}\operatorname{L}_{\mathscr{l}; t}
-\left\|\operatorname{L}_{\mathscr{l}; t}\bm{\psi}_{t}\right\|^{2}\bm{\psi}_{t}\bm{\psi}_{t}^{\dagger}\right )\mathrm{d}t
	\nonumber
\end{align}
We thus arrive at
\begin{align}
	\mathrm{d}\operatorname{E}\big{(}\bm{\psi}_{t}\bm{\psi}_{t}^{\dagger}\big{)}=
	-\imath\,\left[\operatorname{H}_{t}\,,\operatorname{E}\big{(}\bm{\psi}_{t}\bm{\psi}_{t}^{\dagger}\big{)}\right]\mathrm{d}t
	+\sum_{\mathscr{l}=1}^{d^{2}-1}\mathscr{r}_{\mathscr{l}; t}\operatorname{D}_{\operatorname{L}_{\mathscr{l};t}}\left(\operatorname{E}\big{(}\bm{\psi}_{t}\bm{\psi}_{t}^{\dagger}\big{)}\right)\mathrm{d}t
	\nonumber
\end{align}
\end{proof}
The interpretation of  the stochastic process $ \bm{\psi}_{t} $ as state vector is substantiated by the fact that for any initial data
assigned on the Bloch hyper-sphere
\begin{align}
	\left\|\bm{z}\right\|^{2}=1
	\nonumber
\end{align}
we get at any later time
\begin{align}
	\mathrm{d}\left\|\bm{\psi}_{t}\right\|^{2}=0
	\label{LGKS:Bloch}
\end{align}
Namely, the trace the of the stochastic increment
\begin{align}
	\sum_{\mathscr{l}=1}^{d^{2}-1}\mathrm{d}{\nu}_{\mathscr{l};t}\operatorname{Tr}\left(\frac{\operatorname{L}_{\mathscr{l}; t}\bm{\psi}_{t}\bm{\psi}_{t}^{\dagger}\operatorname{L}_{\mathscr{l}; t}}{\left\|\operatorname{L}_{\mathscr{l}; t}\bm{\psi}_{t}\right\|^{2}}-\bm{\psi}_{t}\bm{\psi}_{t}^{\dagger}\right)=\sum_{\mathscr{l}=1}^{d^{2}-1}\mathrm{d}{\nu}_{\mathscr{l};t}
	\left(1-\left\|\bm{\psi}_{t}\right\|^{2}\right)
	\nonumber
\end{align}
vanishes if a jump occurs when the process has unit norm. Similarly, the trace of the drift
\begin{align}
&	\operatorname{Tr}\left (\bm{f}_{t}\bm{\psi}_{t}^{\dagger}+\bm{\psi}_{t}\bm{f}_{t}^{\dagger}\right )=
\nonumber\\
&	-\imath\operatorname{Tr}\left[\operatorname{H}_{t}\,,\bm{\psi}_{t}\bm{\psi}_{t}^{\dagger}\right]
+\sum_{\mathscr{l}=1}^{d^{2}-1}\mathscr{r}_{\mathscr{l};t}\left(\left\|\operatorname{L}_{\mathscr{l}; t}\bm{\psi}_{t}\right\|^{2}-\left\|\operatorname{L}_{\mathscr{l}; t}\bm{\psi}_{t}\right\|^{2}\left\|\bm{\psi}_{t}\right\|^{2}
\right)
	\nonumber
\end{align}
identically vanishes {for realizations of the process with unit norm}. The pathwise preservation of the Bloch hyper-sphere plays an important role for {the interpretation of the stochastic} Schr\"odinger equation as a mathematical model of an indirect measurement record \cite{BaHo1995,WisH1996,BaLu2005}.

\section{Unraveling of the general canonical (completely bounded) master equation}
\label{sec:cb}

We now turn to the main result that we want to illustrate. We distinguish between the completely bounded (\ref{me:me})
and the completely positive (\ref{LGKS:eq}) master equations both being in canonical form -- the conditions (\ref{me:ortho}) hold in both cases --
but differing because the rates $ \mathscr{r}_{\mathscr{l}; t} $'s are positive functions whereas the canonical couplings $  \mathscr{w}_{\mathscr{l}; t}$ may have arbitrary sign. We prove \cite{DoMG2022,DoMG2023}
\begin{theorem}{proposition}
	\label{pp:main}
The solution of the canonical completely bounded master equation (\ref{me:me}) admits the Monte-Carlo representation
\begin{align}
	\bm{\rho}_{t}=\operatorname{E}\left(\mu_{t}\bm{\psi}_{t}\bm{\psi}_{t}^{\dagger}\right)
	\label{cb:main}
\end{align}
where $ \big{\{}\bm{\psi}_{t}\,,\bm{\psi}_{t}^{\dagger}\big{\}}_{t\,\geq\,\ti} $ are stochastic state vectors unraveling a completely positive master equation as in proposition~\ref{prop:LGKS}. The rates of entering the stochastic Schr\"odinger equation (\ref{LGKS:sse}) are related to the canonical couplings of the master equation (\ref{me:cb}) by the unraveling conditions
\begin{align}
	\mathscr{w}_{\mathscr{l}; t}=\mathscr{r}_{\mathscr{l}; t}-\mathscr{c}_{t}
	\label{cb:unraveling}
\end{align}
where  $ \mathscr{c}_{t}\colon\mathbb{R}_{+}\mapsto \mathbb{R}_{+}$ is any real positive function satisfying the condition
\begin{align}
	\mathscr{c}_{t} \,>\, -\min_{\mathscr{l}=1,\dots,d^{2}-1}\mathscr{w}_{\mathscr{l}; t}\,\equiv\,\left |\mathscr{w}_{\mathscr{l}_{\star};  t} \right |
	\label{cb:constraint}
\end{align}
The weighing term $ \mu_{t} $ in the expectation value (\ref{cb:main}) is a scalar martingale -- the \textbf{influence martingale} -- obeying the It\^o stochastic differential equation
\begin{subequations}
	\label{cb:martingale}
	\begin{align}
		&\label{cb:martingale1}	\mathrm{d}\mu_{t}=\mu_{t}\sum_{\mathscr{l}=1}^{d^{2}-1}\left(\frac{\mathscr{w}_{\mathscr{l};t}}{\mathscr{r}_{\mathscr{l};t}}-1\right)\,\mathrm{d}\iota_{\mathscr{l};t}
		\\
		&\label{cb:martingale2}
		\mathrm{d}\iota_{\mathscr{l};t}=	\mathrm{d}\nu_{\mathscr{l};t}-\mathscr{r}_{\mathscr{l};t}\left\|\operatorname{L}_{\mathscr{l}; t}\bm{\psi}_{t}\right\|^{2}\mathrm{d}t
		\\
		&\label{cb:martingale3}\mu_{\ti}=1
	\end{align}
driven by the innovation processes (\ref{cb:martingale2}) (martingale part or compensated increments \cite{KleF2005}) of the increments of counting processes (\ref{LGKS:counting}).
\end{subequations}
\end{theorem}
\begin{proof}~\newline
Owing to (\ref{LGKS:Bloch}) the dynamics preserves the Bloch hyper-sphere. Hence,  for any initial data on the Bloch hyper-sphere, 
the insertion of the unraveling conditions in (\ref{cb:unraveling}) combined with (\ref{me:povm}) leads to the chain of identities
\begin{align}
&	\sum_{\mathscr{l}=1}^{d^{2}-1}\mathscr{r}_{\mathscr{l};t}\frac{\operatorname{L}_{\mathscr{l}; t}^{\dagger}\operatorname{L}_{\mathscr{l}; t}-\left\|\operatorname{L}_{\mathscr{l}; t}\bm{\psi}_{t}\right\|^{2}\operatorname{1}_{d}}{2}\bm{\psi}_{t}=
\nonumber\\
&	\sum_{\mathscr{l}=1}^{d^{2}-1}(\mathscr{w}_{\mathscr{l}; t}+c_{t})\frac{\operatorname{L}_{\mathscr{l}; t}^{\dagger}\operatorname{L}_{\mathscr{l}; t}-\left\|\operatorname{L}_{\mathscr{l}; t}\bm{\psi}_{t}\right\|^{2}\operatorname{1}_{d}}{2}\bm{\psi}_{t}
=	\sum_{\mathscr{l}=1}^{d^{2}-1}\mathscr{w}_{\mathscr{l}; t}\frac{\operatorname{L}_{\mathscr{l}; t}^{\dagger}\operatorname{L}_{\mathscr{l}; t}-\left\|\operatorname{L}_{\mathscr{l}; t}\bm{\psi}_{t}\right\|^{2}\operatorname{1}_{d}}{2}\bm{\psi}_{t}
	\nonumber
\end{align}
The conclusion is that (\ref{LGKS:sse2}) ensures that
\begin{align}
	\bm{f}_{t}=-\imath\,\operatorname{H}_{t}\bm{\psi}_{t}-
	\sum_{\mathscr{l}=1}^{d^{2}-1}\mathscr{w}_{\mathscr{l};t}\frac{\operatorname{L}_{\mathscr{l}; t}^{\dagger}\operatorname{L}_{\mathscr{l}; t}-\left\|\operatorname{L}_{\mathscr{l}; t}\bm{\psi}_{t}\right\|^{2}\operatorname{1}_{d}}{2}\bm{\psi}_{t}
	\label{cb:drift}
\end{align}
Next, we turn to analyze the role of the influence martingale. We apply Ito Lemma again
\begin{align}
	\mathrm{d}(\mu_{t}\bm{\psi}_{t}\bm{\psi}_{t}^{\dagger})=(\mathrm{d}\mu_{t})\bm{\psi}_{t}\bm{\psi}_{t}^{\dagger}+
	\mu_{t}\mathrm{d}(\bm{\psi}_{t}\bm{\psi}_{t}^{\dagger})+(\mathrm{d}\mu_{t})\mathrm{d}(\bm{\psi}_{t}\bm{\psi}_{t}^{\dagger})
	\nonumber
\end{align}
The last term -- { the stochastic quadratic variation}-- changes the intensity of the counting process
\begin{align}
	(\mathrm{d}\mu_{t})\mathrm{d}(\bm{\psi}_{t}\bm{\psi}_{t}^{\dagger})
	=\sum_{\mathscr{l}=1}^{d^{2}-1}\mathrm{d}{\nu}_{\mathscr{l};t}\left(\frac{\mathscr{w}_{\mathscr{l};t}}{\mathscr{r}_{\mathscr{l};t}}-1\right)\mu_{t}\left(\frac{\operatorname{L}_{\mathscr{l}; t}\bm{\psi}_{t}\bm{\psi}_{t}^{\dagger}\operatorname{L}_{\mathscr{l}; t}}{\left\|\operatorname{L}_{\mathscr{l}; t}\bm{\psi}_{t}\right\|^{2}}-\bm{\psi}_{t}\bm{\psi}_{t}^{\dagger}\right)
	\nonumber
\end{align}	
We arrive at
\begin{align}
&	\mathrm{d}(\mu_{t}\bm{\psi}_{t}\bm{\psi}_{t}^{\dagger})=(\mathrm{d}\mu_{t})\bm{\psi}_{t}\bm{\psi}_{t}^{\dagger}+
{\mu_{t}\,\big{(}\bm{f}_{t}\bm{\psi}_{t}^{\dagger}+\bm{\psi}_{t}\bm{f}_{t}^{\dagger}\big{)}}\mathrm{d}t
 \nonumber\\
& +\sum_{\mathscr{l}=1}^{d^{2}-1}\mathrm{d}{\nu}_{\mathscr{l};t}\frac{\mathscr{w}_{\mathscr{l};t}}{\mathscr{r}_{\mathscr{l};t}}\mu_{t}\left(\frac{\operatorname{L}_{\mathscr{l}; t}\bm{\psi}_{t}\bm{\psi}_{t}^{\dagger}\operatorname{L}_{\mathscr{l}; t}}{\left\|\operatorname{L}_{\mathscr{l}; t}\bm{\psi}_{t}\right\|^{2}}-\bm{\psi}_{t}\bm{\psi}_{t}^{\dagger}\right)
	\nonumber
\end{align}
The right hand side non-linearly depends on the state vector process. The non-linearity cancels when averaging over the realizations of the process.
Using again the telescopic property of the conditional expectation we get
\begin{align}
	\operatorname{E}\left((\mathrm{d}\mu_{t})\bm{\psi}_{t}\bm{\psi}_{t}^{\dagger}\right)=
	\operatorname{E}\left(\operatorname{E}\Big{(}\mathrm{d}\mu_{t}\big{|}\big{\{} \bm{\psi}_{t}\,,\bm{\psi}_{t}^{\dagger} \big{\}}\cap\mathfrak{F}_{t}\Big{)}\bm{\psi}_{t}\bm{\psi}_{t}^{\dagger}\right)=0
	\nonumber
\end{align}
because of the martingale property and
\begin{align}
&	\sum_{\mathscr{l}=1}^{d^{2}-1}\operatorname{E}\left(\mathrm{d}{\nu}_{\mathscr{l};t}\frac{\mathscr{w}_{\mathscr{l};t}}{\mathscr{r}_{\mathscr{l};t}}\,\mu_{t}\,\left(\frac{\operatorname{L}_{\mathscr{l}; t}\bm{\psi}_{t}\bm{\psi}_{t}^{\dagger}\operatorname{L}_{\mathscr{l}; t}}{\left\|\operatorname{L}_{\mathscr{l}; t}\bm{\psi}_{t}\right\|^{2}}-\bm{\psi}_{t}\bm{\psi}_{t}^{\dagger}\right)
	\right )
\nonumber\\
&	=\sum_{\mathscr{l}=1}^{d^{2}-1}\frac{\mathscr{w}_{\mathscr{l};t}}{\mathscr{r}_{\mathscr{l};t}}\operatorname{E}\left(\operatorname{E}\Big{(}\mathrm{d}\nu_{\mathscr{l};t}\big{|}\big{\{} \bm{\psi}_{t}\,,\bm{\psi}_{t}^{\dagger} \big{\}}\cap\mathfrak{F}_{t}\Big{)}\,\mu_{t}\,\left(\frac{\operatorname{L}_{\mathscr{l}; t}\bm{\psi}_{t}\bm{\psi}_{t}^{\dagger}\operatorname{L}_{\mathscr{l}; t}}{\left\|\operatorname{L}_{\mathscr{l}; t}\bm{\psi}_{t}\right\|^{2}}-\bm{\psi}_{t}\bm{\psi}_{t}^{\dagger}\right)\right )
	\nonumber\\
&=	\sum_{\mathscr{l}=1}^{d^{2}-1}\mathscr{w}_{\mathscr{l};t}\operatorname{E}\left(\operatorname{L}_{\mathscr{l}; t}\mu_{t}\bm{\psi}_{t}\bm{\psi}_{t}^{\dagger}\operatorname{L}_{\mathscr{l}; t}-\left\|\operatorname{L}_{\mathscr{l}; t}\bm{\psi}_{t}\right\|^{2}\mu_{t}\bm{\psi}_{t}\bm{\psi}_{t}^{\dagger}\right)\mathrm{d}t
\nonumber
\end{align}
Upon recalling the expression of the drift (\ref{cb:drift}) we see that the dependence upon the rates of the state vector process cancels out. We thus have proven that
\begin{align}
	\mathrm{d}\operatorname{E}\left(\mu_{t}\bm{\psi}_{t}\bm{\psi}_{t}^{\dagger}\right)=
	-\imath\,\left[\operatorname{H}_{t}\,,\operatorname{E}\left(\mu_{t}\bm{\psi}_{t}\bm{\psi}_{t}^{\dagger}\right)\right]\mathrm{d}t
	+\sum_{\mathscr{l}=1}^{d^{2}-1}\mathscr{w}_{\mathscr{l}; t}\operatorname{D}_{\operatorname{L}_{\mathscr{l};t}}\left(\operatorname{E}\left(\mu_{t}\bm{\psi}_{t}\bm{\psi}_{t}^{\dagger}\right)\right)\mathrm{d}t
	\nonumber
\end{align}
as claimed.
\end{proof}
The construction of the influence martingale in proposition~\ref{pp:main} differs from the original derivation in \cite{DoMG2022}. There, the proof holds for master equations not necessarily in canonical form: the conditions (\ref{me:ortho}) do not hold by hypothesis and their consequence (\ref{me:povm}) cannot be immediately invoked. {In \cite{DoMG2022}, however,} the drift of the state vector stochastic differential equation satisfies (\ref{cb:drift}) by hypothesis implying that the unraveling conditions (\ref{cb:unraveling}) become dispensable.  The dynamics of the state vector still pathwise preserves the Bloch hyper-sphere.  The drawback is that the average  (\ref{LGKS:main}) does not satisfy a completely positive master equation. As a matter of fact, we show in  \cite{DoMG2023} that the conditions (\ref{me:ortho}) {are sufficient but not necessary} to guarantee that (\ref{LGKS:main}) and (\ref{cb:main}) simultaneously satisfy a completely positive and a completely bounded master equations. Mathematically, master equations are always amenable to canonical form. Physically, a non-canonical form may be preferable to describe actual indirect-measurement experimental setups \cite{BrPe2002,WiMi2009}. 

\subsection{Relation of the influence martingale with the canonical form of completely bounded linear maps.}

From the mathematical point of view, the influence martingale is an extension of the Girsanov martingale for jump processes \cite{KleF2005}. The extension consists of removing the restriction to positive values. Thus, the influence martingale does not describe a change of probability measure. Together, the state vector process, its adjoint dual, and the influence martingale specify a Markov process.  This means that at any instant of time it is always possible to evaluate along a path of the state vector
\begin{align}
	\mu_{t}^{(\pm)}=\max\left(0\,,\pm\mu_{t}\right)	
	\nonumber
\end{align}
and to couch (\ref{cb:main}) into the form
\begin{align}
	\bm{\rho}_{t}=\operatorname{E}\left(\mu_{t}^{(+)}\bm{\psi}_{t}\bm{\psi}_{t}^{\dagger}\right )-\operatorname{E}\left(\mu_{t}^{(-)}\bm{\psi}_{t}\bm{\psi}_{t}^{\dagger}\right)
	\label{unraveling:WP}
\end{align}
	{It is possible to show} \cite{DoMG2022} that each of the two expectations values on the right hand side of (\ref{unraveling:WP}) are always individually amenable to completely positive yet non-trace preserving operator sums. Rather than repeating the analysis of \cite{DoMG2022} in the following section we present an equivalent result using the equation generating ostensible statistics \cite{WisH1996}.

\section{Ostensible statistics}
\label{sec:ostensible}

The stochastic Schr\"odinger equation (\ref{LGKS:sse}) is necessarily non-linear in order to pathwise preserve the Bloch hyper-sphere
\begin{align}
	\left\|\bm{\psi}_{t}\right\|^{2}=1 \hspace{1.0cm}& {\forall\,t\,\geq\,\ti}
	\nonumber
\end{align}
Similarly, the $ \mu_{t} $ needs to be a martingale in order to ensure trace preservation
\begin{align}
	1=\operatorname{Tr}\bm{\rho}_{t}=\operatorname{E}\left (\mu_{t}\,\operatorname{Tr}(\bm{\psi}_{t}\bm{\psi}_{t}^{\dagger})\right )=\operatorname{E}\mu_{t}
	\nonumber
\end{align}
Thus, the need for non-linearity of the state vector evolution and for the martingale property are connected. They both originate if we insist on     
{pathwise} enforcement of probability conservation. {The picture simplifies, if we regard the state vector only as a numerical tool to compute the state operator, the latter being the only physically relevant quantity  as it has also been prominently advocated in \cite{WeiS2014}}. In such a case, the only requirement on the unraveling Markov process is that of generating an ostensible statistics \cite{GaWi2001} whose second order moments average to the solution of the master equation. We refer to \cite{BaPePe2012} for a thorough analysis of the mathematical connection of ostensible statistics with the stochastic Schr\"odinger equation and quantum state diffusion unraveling a completely positive master equation.
We proceed here as in \cite{DoMG2022}: the couplings in (\ref{me:me}) are non positive definite but we do not invoke (\ref{me:povm}). 
We also release the hypothesis that the number of decoherence operators is equal to $ d^{2}-1 $.
We thus consider the linear system of stochastic differential equations for the ostensible state vector
\begin{subequations}
	\label{ost:sde}
	\begin{align}
		&\label{ost:sde1}
		\mathrm{d}\bm{\phi}_{t}=\operatorname{A}_{t}\bm{\phi}_{t}\mathrm{d}t+\sum_{\mathscr{l}=1}^{\mathscr{L}}\mathrm{d}\varpi_{\mathscr{l};  t}\left(\operatorname{L}_{\mathscr{l}; t}\bm{\phi}_{t}-\bm{\phi}_{t}\right)
		\\
		&\label{ost:sde2}
		\mathrm{d}\lambda_{t}= \lambda_{t}\sum_{\mathscr{l}=1}^{\mathscr{L}}\mathrm{d}\varpi_{\mathscr{l}; t}\,\left(\frac{\mathscr{w}_{\mathscr{l}; t}}{\mathscr{r}_{\mathscr{l}; t}}-1\right)
	\end{align}
\end{subequations}
driven by $ \mathscr{L} $ independent Poisson processes specified by 
\begin{subequations}
	\label{ost:Poisson}
	\begin{align}
		&\label{ost:Poisson1}
		\mathrm{d}\varpi_{\mathscr{l};  t}\mathrm{d}\varpi_{\mathscr{k};  t}=\delta_{\mathscr{l},\mathscr{k}}\mathrm{d}\varpi_{\mathscr{k};  t}
		\\
		&\label{ost:Poisson2}
		\operatorname{E}\left(\mathrm{d}\varpi_{\mathscr{l};  t}\big{|}\varpi_{\mathscr{l};  t}\right)=\mathscr{r}_{\mathscr{l}; t}\mathrm{d}t
	\end{align}
\end{subequations}
In (\ref{ost:sde1}) $ \operatorname{A}_{t}\colon \mathbb{R}_{+}\mapsto\mathcal{M}_{d}$ is a time dependent matrix independent of $ \bm{\phi}_{t} $, $ \bm{\phi}_{t}^{\dagger} $ and the Poisson processes. We determine the explicit form of $  \operatorname{A}_{t}$ at the end of the calculation by requiring that 
\begin{align}
	\bm{\rho}_{t}=\operatorname{E}\lambda_{t}\bm{\phi}_{t}\bm{\phi}_{t}^{\dagger}
	\nonumber
\end{align}  
solves a master equation of the form (\ref{me:me}). To (\ref{ost:sde1}) and its dual adjoint we associate initial data $ \bm{z} $, $ \bm{\bar{z}} $ on the Bloch hyper-sphere whereas
{\begin{align}
	\lambda_{\ti}=1
	\nonumber
\end{align}
}
As always, It\^o lemma is our friend:
\begin{align}
\mathrm{d}(\lambda_{t}\bm{\phi}_{t}\bm{\phi}_{t}^{\dagger})&=
	\lambda_{t}(\operatorname{A}_{t}\bm{\phi}_{t}\bm{\phi}_{t}^{\dagger}+\bm{\phi}_{t}\bm{\phi}_{t}^{\dagger}\operatorname{A}^{\dagger})\mathrm{d}t+
	\sum_{\mathscr{l}=1}^{\mathscr{L}}\mathrm{d}\varpi_{\mathscr{l}; t}	\,\left(\frac{\mathscr{w}_{\mathscr{l}; t}}{\mathscr{r}_{\mathscr{l}; t}}-1\right)\,\lambda_{t}\,\bm{\phi}_{t}\bm{\phi}_{t}^{\dagger}
\nonumber\\	
&	+\sum_{\mathscr{l}=1}^{\mathscr{L}}	\frac{\mathscr{w}_{\mathscr{l}; t}}{\mathscr{r}_{\mathscr{l}; t}}\mathrm{d}\varpi_{\mathscr{l}; t}
	\left(\operatorname{L}_{\mathscr{l}; t}\lambda_{t} \bm{\phi}_{t}\bm{\phi}_{t}^{\dagger}\operatorname{L}_{\mathscr{l}; t}^{\dagger}-\lambda_{t}\bm{\phi}_{t}\bm{\phi}_{t}^{\dagger}
	\right)
	\nonumber
\end{align}
Therefore if we  choose
\begin{align}
	\operatorname{A}_{t}=-\imath\operatorname{H}_{t}-\sum_{\mathscr{l}=1}^{\mathscr{L}}
	\frac{\mathscr{w}_{\mathscr{l}; t}\,\operatorname{L}_{\mathscr{l}; t}^{\dagger}\operatorname{L}_{\mathscr{l}; t}-\mathscr{r}_{\mathscr{l}; t}\,\operatorname{1}_{d}}{2}
	\nonumber
\end{align}
we obtain
\begin{align}
	\mathrm{d}(\lambda_{t}\bm{\phi}_{t}\bm{\phi}_{t}^{\dagger})&=-\imath\,\left[\operatorname{H}\,,\lambda_{t}\bm{\psi}_{t}\bm{\psi}_{t}^{\dagger}\right]\,\mathrm{d}t
	\nonumber\\
&	+\sum_{\mathscr{l}=1}^{\mathscr{L}}\mathscr{w}_{\mathscr{l},t}\frac{\left[\operatorname{L}_{\mathscr{l}; t}\,, \lambda_{t}\bm{\phi}_{t}\bm{\phi}_{t}^{\dagger}\operatorname{L}_{\mathscr{l}; t}^{\dagger}\right]+\left[\operatorname{L}_{\mathscr{l}; t}\lambda_{t}\bm{\phi}_{t}\bm{\phi}_{t}^{\dagger}\,, \operatorname{L}_{\mathscr{l}; t}^{\dagger}\right]}{2}
	\nonumber\\
	&+\sum_{\mathscr{l}=1}^{\mathscr{L}}	\left(\mathrm{d}\varpi_{\mathscr{l}; t}-\mathscr{r}_{\mathscr{l}; t}\,\mathrm{d}t\right)
	\left(\frac{\mathscr{w}_{\mathscr{l}; t}}{\mathscr{r}_{\mathscr{l}; t}}\,\operatorname{L}_{\mathscr{l}; t} \,\lambda_{t} \bm{\phi}_{t}\bm{\phi}_{t}^{\dagger}\operatorname{L}_{\mathscr{l}; t}^{\dagger}-\lambda_{t}\bm{\phi}_{t}\bm{\phi}_{t}^{\dagger}\right )
	\label{ost:sLvN}
\end{align}
The sum on the last row runs over the innovation processes (mean compensated increments) of the increments of the Poisson processes.
Hence, upon taking the expectation value of (\ref{ost:sLvN}) we recover (\ref{me:me}).  We may also regard (\ref{ost:sLvN}) as an It\^o stochastic differential equation for the Markov process
\begin{align}
	\bm{\sigma}_{t}=\lambda_{t}\bm{\phi}_{t}\bm{\phi}_{t}^{\dagger}
	\label{ost:cb}
\end{align}
which is, however, trace preserving only on average. It is straightforward now to verify that the expectation value of (\ref{ost:cb}) {takes} the canonical form of a completely bounded linear map. Namely, (\ref{ost:sde1}) is exactly integrable in between jumps. If we introduce 
\begin{align}
&	\operatorname{K}_{t,s}=e^{\sum_{\mathscr{l}=1}^{\mathscr{L}}\int_{s}^{t}\mathrm{d}u\,\frac{\mathscr{r}_{\mathscr{l}; u}}{2}}\operatorname{G}_{t,s}
\nonumber\\
& \operatorname{G}_{t,s}=\mathscr{T}\exp\left(-\int_{s}^{t}\mathrm{d}u\left(\imath\,\operatorname{H}_{u}+\sum_{\mathscr{l}=1}^{\mathscr{L}}\mathscr{w}_{\mathscr{l}; u}{\frac{\operatorname{L}_{\mathscr{l}; u}^{\dagger}\operatorname{L}_{\mathscr{l}; u}}{2}}\right)\right)
\nonumber
\end{align}
where \textquotedblleft $  \mathscr{T}\exp$\textquotedblright\ is the time ordered exponential, then
\begin{align}
	\phi_{\tau_{i+1}^{-}}=\operatorname{K}_{\tau_{i+1}^{-},\tau_{i}^{+}}\phi_{\tau_{i}^{+}}
	\nonumber
\end{align}
having denoted by $ \tau_{i} $ the random time when the $ i $-th jump occurs and by $\phi_{\tau_{i}^{-}}$ ($\phi_{\tau_{i}^{+}}$) the state immediately before (after) the jump. The influence process (\ref{ost:sde2}) is a pure stochastic step function: if the jump 
is caused by the $ \mathscr{l} $-th Poisson process we get
\begin{align}
	\lambda_{\tau_{i}^{+}}=\frac{\mathscr{w}_{\mathscr{l}; \tau_{i}}}{\mathscr{r}_{\mathscr{l}; \tau_{i}}}\lambda_{\tau_{i}^{-}}
	\nonumber
\end{align}
Thus, the change of (\ref{ost:cb}) in a finite time interval including exactly one jump of type $ \mathscr{l}   $ is
\begin{align}
	\sigma_{\tau_{i+1}^{-}}=\frac{\mathscr{w}_{\mathscr{l}; \tau_{i}}}{\mathscr{r}_{\mathscr{l}; \tau_{i}}}
	e^{\sum_{\mathscr{l}=1}^{\mathscr{L}}\int_{\tau_{i}}^{\tau_{i+1}}\mathrm{d}u\,\mathscr{r}_{\mathscr{l}; u}}\operatorname{G}_{\tau_{i+1},\tau_{i}}\operatorname{L}_{\mathscr{l},\tau_{i}}\sigma_{\tau_{i}^{-}}\operatorname{L}_{\mathscr{l},\tau_{i}}^{\dagger}\operatorname{G}_{\tau_{i+1},\tau_{i}}^{\dagger}
	\nonumber
\end{align}
In writing this expression we assumed continuity of canonical rates and couplings. We recognize that the exponential prefactor is exactly the inverse of the probability that no jump occurs in  $\big{[}\tau_{i}^{+}\,,\tau_{i+1}^{-}\big{]}  $: {when evaluating the expectation value it} cancels  with the probability that the trajectory is continuous in between jumps. Hence, in order to compute the expectation value of (\ref{ost:cb}) conditional upon the event that \emph{exactly} one jump occurs  in the interval $ \left[\ti,t\right] $, we only need to recall expression of the 
jump rate (\ref{ost:Poisson2}). The result is   
\begin{align}
	\operatorname{E}(\bm{\sigma}_{t}\big{|}\omega)=\sum_{\mathscr{l}=1}^{\mathscr{L}}\int_{\ti}^{t}\mathrm{d}u\,\mathscr{w}_{\mathscr{l}; u}
	\operatorname{G}_{t,u}\operatorname{L}_{\mathscr{l},u}\operatorname{G}_{u,\ti}\bm{z}\bm{z}^{\dagger}\operatorname{G}_{u,\ti}^{\dagger}\operatorname{L}_{\mathscr{l},u}^{\dagger}\operatorname{G}_{t,u}^{\dagger}
	\nonumber
\end{align}
As to be expected the result is independent of the $\mathscr{r}_{\mathscr{l}; \tau_{i}} $'s.
Inspection of the result shows that the integral can always be couched into the form of a difference between {two} completely positive contributions if one of the canonical couplings $ \mathscr{w}_{\mathscr{l}; \cdot} $ changes sign in $  \left[\ti,t\right]  $. The same considerations can be straightforwardly extended at the price of a somewhat more cumbersome notation to prove that the unconditional expectation value of $ \bm{\sigma}_{t} $ yields a completely bounded map \cite{DoMG2022}. 

\section{Application to a two-level system in a boson environment}
\label{sec:tls}

Gaussian environments linearly coupled to an open system are important mathematical models for multiple reasons. The influence functional \cite{FeVe1963} defined by the exact evaluation of the partial trace constitutes a paradigm of an effective potential felt by an open system. 
The paradigm is also directly descriptive of realistic situations when the correlation time of {the environment} is much faster than the system dynamics \cite{CaLe1983}. 
It is therefore interesting to compare the unraveling generated by the Markov process associated with the influence martingale with that generated by the random process 
 derived for systems in contact with a Gaussian environment  \cite{StrW1996,DiSt1997, DiGiSt1998, StGr2002,StoJ2004,DiFe2014, TilA2017}. For this purpose, we consider here a two level system in a boson environment drawing from \cite{SmVa2010}. There, Smirne and Vacchini derived and analyzed the description of the dynamics given by the master equation and by the Nakajima-Zwanzig memory kernel. Depending upon the number of bosons in the environment the expression of the master equation can be derived exactly or within fourth order time-convolutionless perturbation theory \cite{NeBrWe2021} accuracy. 

In the simplest set-up the environment consists of a single boson either in the vacuum or in the thermal state at the instant of time $ \ti=0 $ when the system environment state operator is assumed in tensor product form. We follow \cite{SmVa2010} and adopt the convention that the first element of the canonical basis of $ \mathbb{C}^{2} $ corresponds to the excited state of the two level states. 
The unitary evolution of the microscopic model is generated by the Hamiltonian operator
\begin{align}
	\operatorname{H}=\omega_{\mathfrak{0}}\,\sigma_{+}\sigma_{-}\,\otimes\,\operatorname{1}_{\mathcal{H}_{E}}+\operatorname{1}_{2}\,\otimes\,(\omega_{\mathfrak{0}}+\Delta)\operatorname{b}^{\dagger}\operatorname{b}+g\left(\sigma_{+}\,\otimes\,\operatorname{b}+\sigma_{-}\,\otimes\,\operatorname{b}^{\dagger}\right)
	\label{tls:H}
\end{align}
where $ \Delta $ is detuning, $ \operatorname{b}\,,\operatorname{b}^{\dagger} $ are the ladder operators in the Hilbert space $ \mathcal{H}_{E} $
of the boson and
\begin{align}
	\sigma_{\pm}=\frac{\sigma_{1}{\pm}\imath\,\sigma_{2}}{2}
	\nonumber
\end{align}
where $ \sigma_{i} $ for $ i=1,2,3 $ are the Pauli matrices. 

The family of matrices in $ \mathcal{M}_{4} $ describing the evolution of the {reshaped} system state operator in the interaction picture has the form
\begin{align}
	\mathsf{F}_{t,0}=\begin{bmatrix}
	\beta_{t} & 0 & 0 & 1-\alpha_{t} \\   0 & \gamma_{t} & 0 & 0\\   0 & 0 &\bar{\gamma}_{t} & 0 \\ \alpha_{t} & 0 & 0 & 1-\beta_{t}
	\end{bmatrix}
	\label{tls:flow}
\end{align} 
The real scalar functions $ \alpha_{t},\beta_{t}\colon\mathbb{R}_{+}\mapsto\mathbb{R} $ and the complex $\gamma_{t}\colon\mathbb{R}_{+}\mapsto\mathbb{C} $ are specified by averages over the initial state operator of environment of certain operators in operators in the Hilbert space of the boson.  In general, any dynamical map obtained by partial trace over unitary evolution from an initial state in tensor product form enjoys complete positivity \cite{BrPe2002}. We determine the conditions that (\ref{tls:flow}) satisfies to ensure complete positivity. The reshuffling of the flow matrix yields
\begin{align}
		\mathsf{F}_{t,0}^{(\mathsf{R})}=\begin{bmatrix}
		\beta_{t} & 0 &  0 & \gamma_{t}\\   0 & 1-\alpha_{t} & 0 & 0\\   0 & 0 &\alpha_{t} & 0\\ \bar{\gamma}_{t} & 0  & 0 & 1-\beta_{t}
	\end{bmatrix}
	\nonumber
\end{align}
The reshuffled matrix is readily self-adjoint with eigenvalues
 \begin{align}
 &	\operatorname{Sp}\mathsf{F}_{t,0}^{(\mathsf{R})}=\left\{f_{1; t},f_{4; t},f_{3; t},f_{4; t}\right\}
 \nonumber\\
&= \left\{ 1-\alpha_{t}\,,1-\beta_{t}\,,\frac{\alpha_{t}+\beta_{t}}{2}+\sqrt{\frac{(\alpha_{t}-\beta_{t})^{2}}{4}+|\gamma_{t}|^{2}}\,, \frac{\alpha_{t}+\beta_{t}}{2}-\sqrt{\frac{(\alpha_{t}-\beta_{t})^{2}}{4}+|\gamma_{t}|^{2}}\right\}
 	\nonumber
 \end{align}
Complete positivity thus requires
\begin{align}
	0\,\leq\,\alpha_{t}\,,\beta_{t}\,\leq\, 1 \hspace{1.0cm}\&\hspace{1.0cm}\alpha_{t}\,\beta_{t}\,\geq\,|\gamma_{t}|^{2}
	\label{tls:cp}
\end{align}

\subsection{Derivation of the completely bounded master equation}

In order to reconstruct the master equation from the {reshaped} dynamical map we observe that
\begin{align}
	{\operatorname{res}}(\bm{\rho}_{t+\mathrm{d}t})=\left(\operatorname{1}_{4}+\dot{\mathsf{F}}_{t}\mathsf{F}_{t}^{-1}\,\mathrm{d}t\right){\operatorname{res}}(\bm{\rho}_{t})
	\nonumber
\end{align}
{It is} sufficient to compute eigenvalues and eigenvectors up to the order $ \mathrm{d}t $ of the reshuffled matrix
\begin{align}
&	\left(\operatorname{1}_{4}+\dot{\mathsf{F}}_{t}\mathsf{F}_{t}^{-1}\,\mathrm{d}t\right)^{\mathsf{R}}=
\nonumber\\
&	\begin{bmatrix}
	1-\mathrm{d}t\frac{\dot{\alpha}_{t}(\beta_{t}-1)-\alpha_{t}\dot{\beta}_{t}}{\alpha_{t}+\beta_{t}-1}	& 0 & 0 &
	1+\mathrm{d}t\left(\operatorname{Re}\frac{\dot{\gamma}_{t}}{\gamma_{t}}+\imath\operatorname{Im}\frac{\dot{\gamma}_{t}}{\gamma_{t}}\right) \\ 
	 0 & \frac{\dot{\beta}_{t}(\alpha_{t}-1)-\beta_{t}\dot{\alpha}_{t}}{\alpha_{t}+\beta_{t}-1}	\mathrm{d}t &0 &0
	 \\
	 0 & 0 & \frac{\dot{\alpha}_{t}(\beta_{t}-1)-\alpha_{t}\dot{\beta}_{t}}{\alpha_{t}+\beta_{t}-1}\mathrm{d}t & 0
	 \\
	 1+\mathrm{d}t\left(\operatorname{Re}\frac{\dot{\gamma}_{t}}{\gamma_{t}}-\imath\operatorname{Im}\frac{\dot{\gamma}_{t}}{\gamma_{t}}\right) & 0 & 0 & 1- \frac{\dot{\beta}_{t}(\alpha_{t}-1)-\beta_{t}\dot{\alpha}_{t}}{\alpha_{t}+\beta_{t}-1}	\mathrm{d}t 
	\end{bmatrix}
	\nonumber
\end{align}
Straightforward algebra then yields
\begin{align}
&	\bm{\dot{\rho}}_{t}=	\frac{\dot{\alpha}_{t}(\beta_{t}-1)-\alpha_{t}\dot{\beta}_{t}}{\alpha_{t}+\beta_{t}-1}		\left(\sigma_{-}\bm{\rho}_{t}\sigma_{+}-\frac{\sigma_{+}\sigma_{-}\bm{\rho}_{t}+\bm{\rho}_{t}\sigma_{+}\sigma_{-}}{2}\right)
\nonumber\\
&+\frac{\dot{\beta}_{t}(\alpha_{t}-1)-\beta_{t}\dot{\alpha}_{t}}{\alpha_{t}+\beta_{t}-1}	\,\left (\sigma_{+}\bm{\rho}_{t}\sigma_{-}-\frac{\sigma_{-}\sigma_{+}\bm{\rho}_{t}+\bm{\rho}_{t}\sigma_{-}\sigma_{+}}{2}\right )
\nonumber\\
&+\frac{1}{4}\left(\frac{\dot{\alpha}_{t}+\dot{\beta}_{t}}{\alpha_{t}+\beta_{t}-1}	-2\operatorname{Re}\frac{\dot{\gamma}_{t}}{\gamma_{t}}\right)\left(\sigma_{3}\bm{\rho}_{t}\sigma_{3}-\bm{\rho}_{t}\right)
+\imath\,\operatorname{Im}\left(\frac{\dot{\gamma}_{t}}{\gamma_{t}}\right)\,\left[\sigma_{+}\sigma_{-}\,,\bm{\rho}_{t}\right]
	\nonumber
\end{align}
If the environment is initially in the vacuum state we get \cite{SmVa2010}
\begin{align}
	\alpha_{t}=1 \hspace{1.0cm} \&  \hspace{1.0cm}\beta_{t}=|\gamma_{t}|^{2}
	\nonumber
\end{align}
with
\begin{align}
 \gamma_{t}=e^{\imath\,\Delta\,t}\left(\cos \frac{t\,\sqrt{\Delta^{2}+2\,g^{2}}}{2}-\imath\frac{\Delta}{\sqrt{\Delta^{2}+4\,g^{2}}}\sin\frac{t\,\sqrt{\Delta^{2}+4\,g^{2}}}{2}\right)
	\nonumber
\end{align}
The explicit form of the master equation becomes
\begin{align}
	\bm{\dot{\rho}}_{t}&=	\frac{2\,g^{2}\,\sqrt{\Delta^{2}+4\,g^{2}}\sin(t\,\sqrt{\Delta^{2}+4\,g^{2}})}{\Delta^{2}+2\,g^{2}\,\left(1+\cos(t\,\sqrt{\Delta^{2}+4\,g^{2}}\right)}\left(\sigma_{-}\bm{\rho}_{t}\sigma_{+}-\frac{\sigma_{+}\sigma_{-}\bm{\rho}_{t}+\bm{\rho}_{t}\sigma_{+}\sigma_{-}}{2}\right)
	\nonumber\\
&	+\imath\,\Delta\frac{\Delta^{2}+2\,g^{2}\,\cos(t\,\sqrt{\Delta^{2}+4\,g^{2}})}{2\,\Delta^{2}+4\,g^{2}\,\left(1+\cos(t\,\sqrt{\Delta^{2}+4\,g^{2}}\right)}\left[\sigma_{+}\sigma_{-}\,,\bm{\rho}_{t}\right]
	\label{eq:master2level}
\end{align}
We see that the above master equation is of the form (\ref{me:me}) because the canonical coupling
\begin{align}
	\mathscr{w}_{t}=\frac{2\,g^{2}\,\sqrt{\Delta^{2}+4\,g^{2}}\sin(t\,\sqrt{\Delta^{2}+4\,g^{2}})}{\Delta^{2}+2\,g^{2}\,\left(1+\cos(t\,\sqrt{\Delta^{2}+4\,g^{2}}\right)}
	\nonumber
\end{align}
changes sign at
\begin{align}
	t_{n}=\frac{\pi+2\,\pi\,n}{\sqrt{\Delta^{2}+4\,g^{2}}}\hspace{1.0cm}\forall\,n\in \mathbb{N}
	\label{tls:divt}
\end{align}

\subsubsection{Analysis at vanishing detuning}

At zero detuning, $ \Delta=0 $, the coupling of the master equation changes sign by passing through infinity:
\begin{align}
	\mathscr{w}_{t}|_{\Delta=0}=2\,g\,\tan(g\,t)
	\nonumber
\end{align}
As it is already evident from inspection of (\ref{tls:flow}), the evolution of the state operator remains well defined:
\begin{align}
	\bm{\rho}_{t}=\frac{\operatorname{1}_{2}-\sigma_{3}}{2}+\cos(g\,t)\,\frac{(x_{1} \sigma_{1}+x_{2} \sigma_{2})+x_{3}\,\cos(g\,t)\sigma_{3}}{2}
	\nonumber
\end{align}
where $ x_{1},x_{2},x_{3}\in\mathbb{R} $ specify the initial data. Singularities of the rate occurs when the state operator reduces to a projector on the ground state. At first glance, the divergence  seems to be a problem for the unraveling. We notice, however, that if we remove the drift in (\ref{ost:sde1}) by setting
\begin{align}
	\bm{\phi}_{t}=e^{\int_{0}^{t}\mathrm{d}s \,\frac{\mathscr{r}_{s}}{2}}\frac{(\cos(g\,t)+1)\operatorname{1}_{2}+(\cos(g\,t)-1)\sigma_{3}}{2}\bm{\xi}_{t}
	\nonumber
\end{align}
then the stochastic process $ \bm{\xi}_{t} $  satisfies the It\^o stochastic jump equation
\begin{align}
\mathrm{d}\bm{\xi}_{t}&=\mathrm{d}\varpi_{t}\left (\frac{(\cos(g\,t)+1)\operatorname{1}_{2}+(1-\cos(g\,t))\sigma_{3}}{2}\sigma_{-}\bm{\xi}_{t}-\bm{\xi}_{t}\right)
\nonumber\\
&	=\mathrm{d}\varpi_{t}\left(\begin{bmatrix}
	0 & 0	\\   \cos(g\,t) & 0
	\end{bmatrix}
-\operatorname{1}_{2}\right)\bm{\xi}_{t}
	\nonumber
\end{align}
whose paths are suppressed when a jump exactly occurs at the times (\ref{tls:divt}) when the canonical rate diverges
\begin{align}
	\forall\bm{\xi}_{t_{n}^{-}}\,\Longrightarrow\, \bm{\xi}_{t_{n}^{+}}=0
	\nonumber
\end{align}
The analysis in this simple case supports the possibility to apply the unraveling even when canonical couplings diverge on a countable set of times.

\subsection{Numerical example}
We numerically integrate the master equation \eqref{eq:master2level} using the non-linear stochastic Schr\"odinger equation and the influence martingale as in section~\ref{sec:cb}. Implementing Bloch hyper-sphere preserving dynamics is numerically convenient in generic situations. It is immediately clear that the single decoherence operator $ \sigma_{-} $ in \eqref{eq:master2level} does not satisfy (\ref{cb:povm1}). As proven in \cite{DoMG2023}, a situation of this type can be easily mended by extending the sums  (\ref{LGKS:sse1}),  (\ref{LGKS:sse1}) to an extra decoherence channel with the following requirements. The addition of the extra decoherence operator to the sum in (\ref{cb:povm1}) must recover an operator proportional to the identity.   The rate of the extra counting process must be exactly equal to $\mathscr{c}_{t}$. In this way by \eqref{cb:unraveling} the decoherence coupling of the extra channel to the master equation equals zero. In the present case, we choose the additional decoherence operator to be $\sigma_{+}$.

Figures \ref{fig:JC}, \ref{fig:JC_2} (a) and (b) show the results of the numerical integration of the influence martingale for 10000 and 20000 trajectories for two different initial states. The coloured (full) lines show the trace of the state and expectation values normalised by the trace. The results of the influence martingale are compared with the direct integration of the master equation, which is displayed by the dashed lines. We observe that doubling the amount of simulated trajectories improves the estimation of the solution to the master equation.

Numerically integration of master equations using the influence martingale is currently being implemented in Qutip version 5 \cite{JoNaNo2013,Qutip}
\begin{figure}
    \centering
    \includegraphics[scale=0.8]{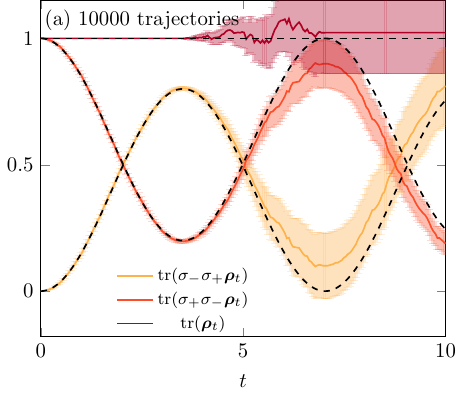}
    \includegraphics[scale=0.8]{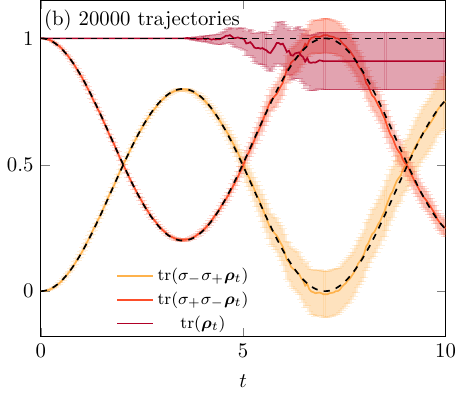}
    \caption{Numerical integration of master equation \eqref{eq:master2level} using the influence martingale. The full coloured lines show the estimations from the influence martingale and the shaded regions show twice the estimated standard deviation divided by the square root of the amount of trajectories. The dashed black lines show the results of directly integrating the master equation. The system parameters are $\Delta=0.4$ and $g=0.4$. The system initial state equals $|-\rangle$ with $\sigma_z |-\rangle=-|-\rangle$.}
    \label{fig:JC}
\end{figure}

\begin{figure}
    \centering
    \includegraphics[scale=0.8]{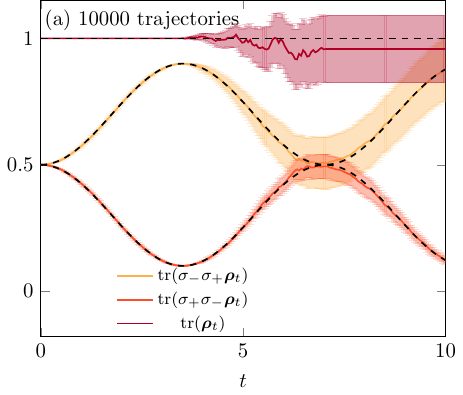}
    \includegraphics[scale=0.8]{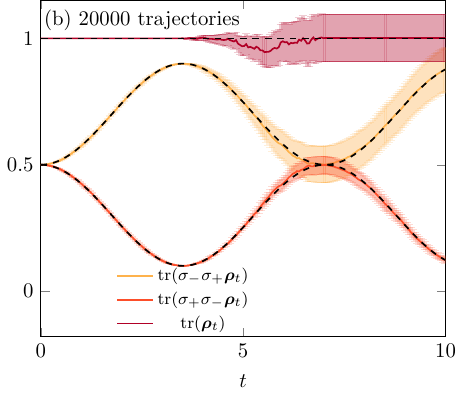}
    \caption{Numerical integration of master equation \eqref{eq:master2level} using the influence martingale. The full coloured lines show the estimations from the influence martingale and the shaded regions show twice the estimated standard deviation divided by the square root of the amount of trajectories. The dashed black lines show the results of directly integrating the master equation. The system parameters are $\Delta=0.4$ and $g=0.4$. The system initial state equals $\frac{|-\rangle+ |+\rangle}{\sqrt{2}}$ with $\sigma_z |\pm\rangle=\pm|\pm\rangle$.}
    \label{fig:JC_2}
\end{figure}

\subsubsection{Remark on the exponential growth of the influence martingale}

Upon decomposing the influence martingale into compensator and pure jump parts \cite{KleF2005} we can always write
\begin{align}
	\label{eq:sigma}
	\mu_t = \exp\left(\mathscr{g}\int_0^t \mathrm{d}s \, \mathscr{c}_{s}\right)\gamma_t
\end{align}
with
\begin{align}
	\mathrm{d}\gamma_t =\gamma_{t} \sum_{\mathscr{l}=1}^{d^{2}-1}\left(\frac{\mathscr{w}_{\mathscr{l};t}}{\mathscr{r}_{\mathscr{l};t}}-1\right)\,\mathrm{d}\nu_{\mathscr{l};t}.
\end{align}
We readily see that the presence of at least one negative coupling brings about a deterministic exponential growth of the martingale. Yet, the presence of this exponential prefactor is a consequence of our choice of representation of the unraveling rather than an a property of the master equation. In fact, we can take advantage of the chain of identities 
\begin{align}
	\bm{\rho}_{t}=\operatorname{E}(\mu_{t}\bm{\psi}_{t}\bm{\psi}_{t}^{\dagger})=\frac{\operatorname{E}(\mu_{t}\bm{\psi}_{t}\bm{\psi}_{t}^{\dagger})}{\operatorname{E}(\mu_{t})}=\frac{\operatorname{E}(\gamma_{t}\bm{\psi}_{t}\bm{\psi}_{t}^{\dagger})}{\operatorname{E}(\gamma_{t})}
	\nonumber
\end{align}
and use the rightmost expression in numerical estimates of the state operator. Thus, the exponential growth of $\mu_t$ does not pose a problem.

\section{Di\'osi-Gisin-Strunz random (ostensible) state vector equation}
\label{sec:DGS}

In a series of remarkable papers  Strunz \cite{StrW1996} in collaboration with Di\'osi and Gisin \cite{DiSt1997,DiGiSt1998} 
showed that the state operator of an open system in contact with a Gaussian environment admits the representation as the expectation 
value over realizations of ostensible pure states solution of a random differential equation driven by time correlated complex-valued Gaussian processes. In \cite{DiGiSt1998} via Girsanov's change of measure method they also derived of a non-linear norm preserving dynamics on the Bloch hyper-sphere. The derivation of Gaussian random differential equation has been later refined in \cite{DiFe2014,LiSt2017} and \cite{StGr2002,StoJ2004,TilA2017} where some essential improvement have been introduced that greatly simplify actual applications of the method.

We analyze here the ostensible linear dynamics, which is the backbone of the method. As our purpose is comparison with the influence martingale approach, we focus here on the simplest possible example which captures the general features of the method. We refer to \cite{StGr2002,StoJ2004,LiSt2017} and especially \cite{TilA2017} for the general case.

\subsection{Influence functional of a simple boson environment}

We consider an arbitrary system on a $ d $ dimensional Hilbert space coupled by a decoherence operator $ \operatorname{L} $ and its adjoint to the normal modes of the {boson} environment  
\begin{align}
	\operatorname{H}_{SE}=\operatorname{H}\,\otimes\,\operatorname{1}_{E}+\sum_{n=1}^{\mathscr{N}} (g_{n}\operatorname{L}\,\otimes\,\operatorname{a}_{n}^{\dagger}
	+\bar{g}_{n}\operatorname{L}^{\dagger}\,\otimes\,\operatorname{a}_{n})+\operatorname{1}_{S}\,\otimes\,\sum_{n=1}^{\mathscr{N}}
	\omega_{n}\operatorname{a}_{n}^{\dagger}\operatorname{a}_{n}
	\label{DGS:micro}
\end{align}
In what follows we make the simplifying but not too restrictive assumption that at $ t=0 $ the state operator of the bipartite system-environment {is the} tensor product of the system state operator and of the state operators of the individual normal modes of the environment.

The hypothesis of linear coupling has the consequence that in the path integral representation system operators act as external sources. 
In the holomorphic (also referred to as Keldysh's) representation \cite{ZiJu2010,KamA2011} the influence functional 
\cite{FeVe1963}\footnote{We refer to \cite{SuChYuCh1988} for the proof of the equivalence of the holomorphic influence functional and its the position-momentum representation first introduced in \cite{FeVe1963}. We thank Erik Aurell for commenting on this point.} in a time interval $ [0,\tf] $ corresponds to the path integral 
\begin{align}
	\mathcal{Z}(\bm{\mathscr{j}},\bm{\bar{\mathscr{j}}})=\int_{\mathfrak{P}}\prod_{n=1}^{\mathscr{N}}\operatorname{D}[\bm{\mathscr{a}}_{n},\bm{\bar{\mathscr{a}}}_{n}]\, e^{
	\mathcal{A}_{n; \tf}(\bm{\mathscr{a}}_{n},\bm{\bar{\mathscr{a}}}_{n},{\bm{\mathscr{j}},\bm{\bar{\mathscr{j}}}})}
	\label{DGS:pi}
\end{align} 
The action functional weighing each of the normal modes is
\begin{align}
{	\mathcal{A}_{n; \tf}(\bm{\mathscr{a}},\bm{\bar{\mathscr{a}}},\bm{\mathscr{j}},\bm{\bar{\mathscr{j}}})
=\imath\,\int_{0}^{\tf}\mathrm{d}t\,\Big{(}\bm{\bar{\mathscr{a}}}_{t}^{\top}
\sigma_{3}(\imath\,\partial_{t}-\omega_{n})\bm{\mathscr{a}}_{t}
-\bar{g}_{n}\bm{\bar{\mathscr{j}}}_{t}^{\top}\sigma_{3}\bm{\mathscr{a}}_{t}-g_{n}\bm{\bar{\mathscr{a}}}_{t}^{\top}\sigma_{3}\bm{\mathscr{j}}_{t})\Big{)}
}
	\nonumber
\end{align}
where the coordinates of an individual normal mode and source terms associated to the decoherence operators are $ \mathbb{C}^{2} $-valued vectors
\begin{align}
	\bm{\mathscr{a}}_{t}
	=\begin{bmatrix}
	\mathscr{a}_{\mathfrak{l}; t}	\\  \mathscr{a}_{\mathfrak{r}; t}
	\end{bmatrix}
, \hspace{0.3cm}
	\bm{\bar{\mathscr{a}}}_{t}
	=\begin{bmatrix}
		\mathscr{\bar{a}}_{\mathfrak{l}; t}	\\  \mathscr{\bar{a}}_{\mathfrak{r}; t}
	\end{bmatrix}
	, \hspace{0.3cm}
	\bm{\mathscr{j}}_{ t}
	=\begin{bmatrix}
	\mathscr{j}_{\mathfrak{l}; t}	\\  \mathscr{j}_{\mathfrak{r}; t}
\end{bmatrix}	, \hspace{0.3cm}
\bm{\bar{\mathscr{j}}}_{ t}
=\begin{bmatrix}
\mathscr{\bar{j}}_{\mathfrak{l}; t}	\\  \mathscr{\bar{j}}_{\mathfrak{r}; t}
\end{bmatrix}
	\nonumber
\end{align}
and $ \sigma_{i} $, $ i=1,2,3 $ here and below are the Pauli matrices.
Physically,  the vectors $ \bm{\mathscr{a}}_{t} $, $ \bm{\bar{\mathscr{a}}}_{t} $ correspond to  ladder operators of a normal mode: the $ \mathfrak{l} $-labeled ($ \mathfrak{r} $-labeled) component corresponds to an operator acting from the left (right) of the state operator.  Mathematically, $ \bm{\mathscr{a}}_{t}\,, \bm{\bar{\mathscr{a}}}_{t}$  are independent integration variables in $ \mathbb{C}^{2} $
(see e.g. discussion in section~6.1.1 of \cite{ZiJu2010}). To emphasize this point and later use, we denote index vector contraction with the algebraic transposition symbol ``$ \top $''.

In (\ref{DGS:pi}), $ \mathfrak{P} $ means that the path-space is restricted by the boundary conditions
\begin{align}
\mathfrak{P}=\cup_{n=1}^{\mathscr{N}} \left\{ \big{\{}\bm{\mathscr{a}}_{n; t},\bm{\bar{\mathscr{a}}}_{n; t}\big{\}}_{t\in [0,\tf]} \,\big{|} \,
\mathsf{A}_{n; \iota}\bm{\mathscr{a}}_{n; 0}+ \mathsf{A}_{n;\mathscr{f}}\bm{\mathscr{a}}_{n; \tf} =0\right\}
\nonumber
\end{align}
which are most conveniently imposed by introducing a pair of matrices specifying the linear subspace in $ \mathbb{C}^{2} $ transverse to the boundary conditions. In the case each of the normal modes is initially in a thermal state at temperature $ \beta^{-1} $ we get
\begin{align}
	\label{DGS:Bott}
	\mathsf{A}_{n; \iota}=\begin{bmatrix}
		0  & 0	\\  1 & -\,e^{-\beta\,\omega_{n}}
	\end{bmatrix}
	\hspace{1.0cm}\&\hspace{1.0cm}
	\mathsf{A}_{n; \mathscr{f}}=\begin{bmatrix}
		1 & -1	\\ 0 & 0 
	\end{bmatrix}		
\end{align}
Strictly speaking (\ref{DGS:pi}) should be regarded as a formal expression whose precise definition is given by the limit of finite dimensional approximations on a time lattice. We refer to chapter~2 of \cite{KamA2011} for a detailed discussion of the mathematical subtleties that may arise. Upon evaluating the Gaussian path integrals and applying functional determinant theory   (see e.g. \cite{PMG2003}) arrive at the explicit expression of the influence functional
{\begin{align}
	\mathcal{Z}(\bm{\mathscr{j}},\bm{\bar{\mathscr{j}}})=
	\prod_{n=1}^{\mathscr{N}} 
	\frac{\exp\left(-\imath g_{n}^{2}\int_{0}^{\tf}\mathrm{d}t\,\int_{0}^{\tf}\mathrm{d}s\,\bm{\bar{\mathscr{j}}}_{t}^{\top}\sigma_{3}
		\bm{\mathscr{G}}_{n; t,s}\sigma_{3}\bm{\mathscr{j}}_{ s}\right)}{\det\left(\mathsf{A}_{n;\mathscr{f}}\bm{\mathscr{F}}_{n; \tf,0}+\mathsf{A}_{n; \iota}\right)} 
	\label{DGS:if}
\end{align}
}
with $ \bm{\mathscr{G}}_{n; t,s}$ the Green function solution of
\begin{subequations}
	\label{DGS:Green}
	\begin{align}
		&\label{DGS:Green1}
		\sigma_{3}\,(\imath\,\partial_{t}-\omega_{n})\bm{\mathscr{G}}_{n; t,s}=\delta_{t-s}\operatorname{1}_{2}&
		\\
		&\label{DGS:Green2}
		\mathsf{A}_{n; \mathscr{f}}\bm{\mathscr{G}}_{n; \tf,s}+\mathsf{A}_{n; \iota}\bm{\mathscr{G}}_{n; 0,s}=0&\forall\, s \in \left(0,\tf\right)
	\end{align}
\end{subequations}
and $ \bm{\mathscr{F}}_{n; t,s}$ solving
\begin{subequations}
	\label{DGS:flow}
	\begin{align}
		&\label{DGS:flow1}
		(\partial_{t}-\imath\,\omega_{n})\bm{\mathscr{F}}_{n; t,s}=0&
		\\
		&\label{DGS:flow2}
		\bm{\mathscr{F}}_{n; t\,s}=\operatorname{1}_{2}&\forall\, s \in \left(0,t\right)
	\end{align}
\end{subequations}
 The determinant {in the denominator of} (\ref{DGS:if}) is a pure function of time which we can reabsorb in the normalization of the path integral over system degrees of freedom. Thus, all the relevant information about the system environment interaction is subsumed in the kernel \cite{FeVe1963}
\begin{align}
\bm{\mathscr{K}}_{t,s}:=-\imath\,\sum_{n=1}^{\mathscr{N}}|g_{n}|^{2}
	\sigma_{3}\bm{\mathscr{G}}_{n; t,s}\sigma_{3}=
	{	
	\begin{bmatrix}
	\mathscr{K}_{ t,s}^{(1,1)} & \mathscr{K}_{ t,s}^{(1,2)} 	\\  \mathscr{K}_{ t,s}^{(2,1)}  & \mathscr{K}_{ t,s}^{(2,2)}
	\end{bmatrix}
}
	\label{DGS:quad}
\end{align}
Once we insert the explicit solution of (\ref{DGS:Green}), we get
\begin{align}
	\bm{\mathscr{K}}_{t,s}=
		\begin{bmatrix}
			-\mathds{1}_{t-s}^{(0)}\,f_{t,s}^{(1)}-\mathds{1}_{s-t}^{(1)}\,f_{t,s}^{(2)} & f_{t,s}^{(2)} 
			\\ 
			f_{t,s}^{(1)} & -\mathds{1}_{t-s}^{(1)}\,f_{t,s}^{(2)}-\mathds{1}_{s-t}^{(0)}\,f_{t,s}^{(1)}
			\end{bmatrix}
	\label{DGS:covariance1}
\end{align}
with
\begin{align}
	\label{DGS:covfun}
	\begin{array}{ll}
	f_{t,s}^{(1)}=\sum_{n=1}^{\mathscr{N}}e^{-\imath\,\omega_{n}\,(t-s)}\dfrac{g_{n}^{2} }{1-e^{-\beta\omega_{n}}}=\bar{f}_{s,t}^{(1)}=f_{t-s,0}^{(1)}
		\\[0.3cm]  
	f_{t,s}^{(2)}=\sum_{n=1}^{\mathscr{N}}e^{-\imath\,\omega_{n}\,(t-s)}\dfrac{g_{n}^{2}	e^{-\beta\,\omega_{n}}}{1-e^{-\beta\omega_{n}}}=\bar{f}_{s,t}^{(2)}=f_{t-s,0}^{(2)}
	\end{array}
\end{align}
The value of the Heaviside functions $ \mathds{1}_{t}^{(x)} $for zero argument $ \mathds{1}_{0}^{(x)}=x $  is fixed  to the identity by evaluating the path integral from finite dimensional approximations \cite{KamA2011}. 

\subsection{Construction of effective complex Gaussian processes}

We now would like to relate (\ref{DGS:covariance1}) to the covariance of zero mean complex-valued Gaussian processes.  To this goal, it expedient
to write the sources as
\begin{align}
	\begin{array}{ll}
	\mathscr{j}_{\mathfrak{i}; t}=\dfrac{\mathscr{r}_{\mathfrak{i}; t}^{(1)}+\imath\,\mathscr{r}_{\mathfrak{i}; t}^{(2)}}{\sqrt{2}}
\\[0.3cm]
\bar{\mathscr{j}}_{\mathfrak{i}; t}=\dfrac{\mathscr{r}_{\mathfrak{i}; t}^{(1)}-\imath\,{\mathscr{r}_{\mathfrak{i}; t}^{(2)}}}{\sqrt{2}}
	\end{array}
\hspace{1.0cm}\mathfrak{i}=\mathfrak{l},\mathfrak{r}
	\nonumber
\end{align} 
The representation reflects the fact that any operator can be decomposed in two self-adjoint components. Correspondingly, we introduce the vectors
\begin{align}
	\bm{\mathscr{r}}_{t}^{(i)}=\begin{bmatrix}
	\mathscr{r}_{\mathfrak{l}; t}^{(i)}	\\  \mathscr{r}_{\mathfrak{r}; t}^{(i)}	
	\end{bmatrix} \in \mathbb{R}^{2},
\hspace{0.5cm}
i=1,2
\hspace{1.0cm}
{\bm{\mathscr{s}}_{t}=\begin{bmatrix}
	\bm{\mathscr{r}}_{t}^{(1)}	\\  \bm{\mathscr{r}}_{t}^{(2)}
\end{bmatrix}}
\in \mathbb{R}^{4}
	\nonumber
\end{align}
We can now write the quadratic form in the influence functional as the sum of three contributions:
\begin{align}
&	\int_{0}^{\tf}\mathrm{d}t\,\int_{0}^{\tf}\mathrm{d}s\,{\bm{\bar{\mathscr{j}}}_{t}^{\top} \bm{\mathscr{K}}_{t,s}\bm{\mathscr{j}}_{ s}}
	=
\nonumber\\	
&	\sum_{i=1}^{2}\int_{0}^{\tf}\mathrm{d}t\,\int_{0}^{\tf}\mathrm{d}s\,\bm{\mathscr{r}}_{t}^{(i)\top} \bm{\mathscr{K}}_{t,s}^{\mathsf{S}}\bm{\mathscr{r}}_{ s}^{(i)}
	+\int_{0}^{\tf}\mathrm{d}t\,\int_{0}^{\tf}\mathrm{d}s\,\bm{\mathscr{s}}_{t}^{\top} \bm{\mathscr{I}}_{t,s}\bm{\mathscr{s}}_{ s}
	\label{DGS:effective}
\end{align}
The first two integrals specify symmetric quadratic forms involving each a distinct two component vector. Both quadratic forms  depend upon the symmetrized kernel 
\begin{align}
	\bm{\mathscr{K}}_{t,s}^{\mathsf{S}}=\frac{\bm{\mathscr{K}}_{t,s}+\bm{\mathscr{K}}_{s,t}^{\top}}{2}
	\nonumber
\end{align}
 The two quadratic forms describe the interaction of the system with the environment mediated by { the self-adjoint component of the Toeplitz decomposition of the decoherence operator $ \operatorname{L} $}. The third integral describes interference and is therefore associated to a symmetric quadratic form involving a single four component vector.

The representation (\ref{DGS:effective}) paves the way for unraveling the influence functional.  To get there a little more algebraic analysis is needed.
We observe that the diagonal matrix elements of (\ref{DGS:covariance1}) transform under complex conjugation as
\begin{align}
	\bar{\mathscr{K}}_{ t,s}^{(1,1)}=\mathscr{K}_{ s,t}^{(2,2)}
	\nonumber
\end{align}
It is thus expedient to define
\begin{align}
	& f_{t,s}=\frac{\bar{f}_{t,s}^{(1)}+f_{t,s}^{(2)}}{2}={\bar{f}_{s,t}}
	\nonumber\\
	&
	\mathscr{S}_{ t,s}=\operatorname{Re}\left(\dfrac{\mathscr{K}_{ t,s}^{(1,1)}+\bar{\mathscr{K}}_{ t,s}^{(2,2)}}{2}\right) =\mathscr{S}_{ s,t}
	\nonumber\\
	&\mathscr{A}_{ t,s}=\imath\operatorname{Im}\left(\dfrac{\mathscr{K}_{ t,s}^{(1,1)}+\bar{\mathscr{K}}_{ t,s}^{(2,2)}}{2}   \right)
	= \mathscr{A}_{ s,t}
	\nonumber
\end{align}
and  to decompose the symmetrized kernel as
\begin{align}
\bm{\mathscr{K}}_{ t,s}^{\mathsf{S}}=\operatorname{Z}_{t,s}\sigma_{1} 
	+\operatorname{R}_{t,s}
	\nonumber
\end{align}
where $ \sigma_{1} $ is as before the first Pauli matrix and
\begin{align}
	\operatorname{Z}_{t,s}=\begin{bmatrix}
		f_{t,s} & -\mathscr{S}_{t,s}
		\\  
		-\mathscr{S}_{t,s}  & \bar{f}_{t,s}
	\end{bmatrix}=\operatorname{Z}_{s,t}^{\dagger}
		\label{DGS:popcov}
\end{align}
and
\begin{align}
\operatorname{R}_{t,s}=	\begin{bmatrix}
		-\mathscr{A}_{t,s}	 &  0 \\  0 &	\mathscr{A}_{t,s}
	\end{bmatrix}=\operatorname{R}_{s,t}^{\top}=\sigma_{1}\bar{\operatorname{R}}_{t,s}\sigma_{1}
	\label{DGS:compcov}
\end{align}
Thus, (\ref{DGS:popcov}) and (\ref{DGS:compcov})  respectively transform as the proper and complementary covariance of a complex Gaussian problem. The last condition that we need to verify is that the augmented covariance constructed using (\ref{DGS:popcov}) and (\ref{DGS:compcov}) 
as blocks (see  formula (\ref{cGp:block}) of Appendix~\ref{ap:cGp}) is semi-positive definite. This is generically the case based on (\ref{DGS:covfun}). We can therefore write  for $ i=1,2 $
\begin{align}
&	\int_{0}^{\tf}\mathrm{d}t\,\int_{0}^{\tf}\mathrm{d}s\,\bm{\mathscr{r}}_{t}^{(i)\top} \bm{\mathscr{K}}_{t,s}^{\mathsf{S}}\bm{\mathscr{r}}_{ s}^{(i)}=
	\nonumber\\
&
\int_{0}^{\tf}\mathrm{d}t\,\int_{0}^{\tf}\mathrm{d}s\,\left (\bm{\mathscr{r}}_{t}^{(i)\top} \operatorname{Z}_{t,s}\sigma_{1}{\bm{\mathscr{r}}_{s}^{(i)}}
+\bm{\mathscr{r}}_{t}^{(i)\top}\frac{\operatorname{R}_{t,s}+\sigma_{1}\bar{\operatorname{R}}_{t,s}\sigma_{1}}{2}{\bm{\mathscr{r}}_{s}^{(i)}}\right )
	\label{DGS:cGpdec}
\end{align}
We recognize (see Appendix~\ref{ap:cGp}) that this is the logarithm of the generating function of a complex Gaussian process. In the absence of interference, i.e. when the interaction with the environment is mediated by a single self-adjoint decoherence operator this is the only contribution
\cite{StGr2002}.
Turning to the interference term, the kernel of the quadratic form has the bock structure
\begin{align}
	\bm{\mathscr{I}}_{t,s}={\frac{1}{4\,\imath}\begin{bmatrix}
	0 & \operatorname{I}_{t,s}	\\   {\operatorname{I}_{s,t}^{\top}} & 0
	\end{bmatrix}},
\hspace{1.0cm}
\operatorname{I}_{t,s}={\begin{bmatrix}
	\mathscr{K}_{ t,s}^{(1,1)}-\bar{\mathscr{K}}_{ t,s}^{(2,2)}	& -\mathscr{K}_{ t,s}^{(1,2)}+\bar{\mathscr{K}}_{ t,s}^{(2,1)}
	\\
\bar{\mathscr{K}}_{ t,s}^{(1,2)}  -\mathscr{K}_{ t,s}^{(2,1)}&  -\bar{\mathscr{K}}_{ t,s}^{(1,1)}+\mathscr{K}_{ t,s}^{(2,2)}
\end{bmatrix}}
	\nonumber
\end{align}
 and therefore enjoys the properties of a complementary covariance of an improper complex Gaussian process: 
\begin{align}
	\bm{\mathscr{I}}_{t,s}=\bm{\mathscr{I}}_{s,t}^{\top}
	\nonumber
\end{align}
It is therefore legitimate the interpretation of the interference term as the logarithm of half the generating function of a {complex} Gaussian process.
The upshot is that we can couched the influence functional  into the form
\begin{align}
	\exp\left(\int_{0}^{\tf}\mathrm{d}t\,\int_{0}^{\tf}\mathrm{d}s\,\bm{\mathscr{j}}_{t}^{\top}
	\bm{\mathscr{K}}_{ t,s}\bm{\mathscr{j}}_{ s}\right)=\prod_{i=1}^{2}{\mathcal{E}}_{\mathcal{G}_{i}}(\bm{\mathscr{r}}^{(i)}, \sigma_{1}\bm{\mathscr{r}}^{(i)})\mathcal{E}_{\mathcal{G}_{3}}(\bm{\mathscr{s}},0)
	\nonumber
\end{align}
with $ {	\mathcal{E}}_{\mathcal{G}_{i}} $, $ i=1,2 $ the generating functions measures of the $ \mathbb{C}^{2} $-valued independent Gaussian processes $ \bm{\zeta}=\big{\{}\bm{\zeta}_{t}^{(i)}\big{\}}_{t\,\geq\,0} $,   $ i=1,2 $
\begin{align}
	&	{	\mathcal{E}}_{\mathcal{G}_{i}}(\bm{\mathscr{r}}^{(i)}, \sigma_{1}\bm{\mathscr{r}}^{(i)}):=\operatorname{E}_{\mathcal{G}_{i}}\exp\left(\int_{0}^{t}\mathrm{d}s \,\Big{(}\bm{\mathscr{r}}_{s}^{(i)\top}\bm{\zeta}_{s}^{(i)}+\bm{\mathscr{r}}_{s}^{\top(i)}\sigma_{1}\bm{\bar{\zeta}}_{s}^{(i)}\Big{)}\right)
	\nonumber
\end{align}
and of the the $ \mathbb{C}^{4} $-valued Gaussian process $ \bm{\eta}=\big{\{}\bm{\eta}_{t}\big{\}}_{t\,\geq\,0} $ also independent from the other ones:
\begin{align}
{	\mathcal{E}}_{\mathcal{G}_{3}}(\bm{\mathscr{\sigma}},0):=	{\operatorname{E}_{\mathcal{G}_{3}}}\exp\left(\int_{0}^{t}\mathrm{d}s \,\bm{\mathscr{\sigma}}_{s}^{\top}\bm{\eta}_{s}\right)
	\nonumber
\end{align}

\subsection{Ostensible state vector equations}

The construction of the complex Gaussian process comes at a price. The state operator is obtained from the expectation value
\begin{align}
	\bm{\rho}_{t}=\operatorname{E}_{\mathcal{G}_{1},\mathcal{G}_{2},\mathcal{G}_{3}}\bm{\varphi}_{t}\bm{\phi}_{t}^{\top}
	\label{DGS:unraveling}
\end{align}
of the outer product of {solutions of random differential equations not related by the adjoint operation} \cite{TilA2017}
\begin{align}
	&		\bm{\dot{\varphi}}_{t}=-\imath\,\operatorname{H}\bm{\varphi}_{t}+\frac{\operatorname{L}+\operatorname{L}^{\dagger}}{\sqrt{2}}(\gamma_{1; t}+\eta_{1; t})\bm{\varphi}_{t}
	+\frac{\operatorname{L}-\operatorname{L}^{\dagger}}{\sqrt{2}\,\imath}(\gamma_{2; t}+\eta_{3; t})\bm{\varphi}_{t}
	\nonumber\\
	&	\bm{\dot{\phi}}_{t}=\imath\,\operatorname{H}\bm{\phi}_{t}+\frac{\operatorname{L}+\operatorname{L}^{\dagger}}{\sqrt{2}}(\bar{\gamma}_{1; t}+\eta_{2; t})\bm{\phi}_{t}+\frac{\operatorname{L}-\operatorname{L}^{\dagger}}{\sqrt{2}\,\imath}(\bar{\gamma}_{2; t}+\eta_{4; t})\bm{\phi}_{t}
	\nonumber
\end{align}
with $ \bm{\gamma}=\big{\{}\bm{\gamma}_{t}\big{\}}_{t\,\geq\,0} $ {the} complex-valued Gaussian process in $ \mathbb{C}^{2} $ whose component are specified by those of the Gaussian processes introduced in the unraveling:
\begin{align}
	&	\gamma_{1; t}=\zeta_{1; t}^{(1)}+\bar{\zeta}_{2; t}^{(1)}
	\nonumber\\
	&\gamma_{2; t}=\zeta_{1; t}^{(2)}+\bar{\zeta}_{2; t}^{(2)}
	\nonumber
\end{align}
We emphasize that the proper covariance of the Gaussian process $ \bm{\eta} $ does not play any role for the evaluation of the state operator.  
We also notice that if in non-generic cases the augmented covariance (\ref{DGS:cGpdec}) is not immediately semi-positive definite, it can be always regularized by adding to the unraveling of $ \bm{\varphi} $ and $ \bm{\phi} $ a extra proper complex Gaussian process with imaginary prefactor. As the two processes are not related by dual adjoint, the result is that dynamical correlations are independent of the regularization.

{The result (\ref{DGS:unraveling}) that the state operator is the expectation value of two vectors one not the adjoint dual of the other is not too surprising.  At arbitrary coupling the canonical master equation is of the form (\ref{me:me}) with couplings of arbitrary sign. It therefore generates a completely bounded semi-group of linear maps. It is always possible to couch a completely bounded linear map as the off-diagonal block of a completely positive map on operators 
in an embedding Hilbert space. Based on this observation, \cite{BrKaPe1999,BrPe2002,BreH2004} showed that solutions of (\ref{me:me}) admit  a representation of the form (\ref{DGS:unraveling}) with two vectors such that their direct sum is solution of a stochastic Schr\"odinger equation.}

Finally, a straightforward application of Gaussian integration by parts \cite{NuaD1995}, permits to write the equation for the state operator in terms of functional derivatives of the ostensible vectors, hence giving a precise mathematical meaning to the equations originally found in \cite{DiSt1997} see also \cite{DiFe2014}.

\section{Comparison of the unravelings}
\label{sec:comparison}

The foregoing analysis of the Gaussian random differential equation highlights merits and differences with the unraveling via the influence martingale 
\cite{DoMG2022,DoMG2023}. 

The most evident difference is that the influence martingale is an all-encompassing existence result of the unraveling of open quantum systems. Existence is granted under the same generic conditions which guarantee that of the canonical master equation.  The unraveling is naturally formulated in terms of a stochastic state vector evolving on the system Bloch hyper-sphere according to It\^o ordinary stochastic differential equations. The evolution is therefore inherently non-anticipating of collapse events. The fact makes available a mature set of numerical algorithms (see e.g. \cite{PlBrLi2010}) and concentration estimates \cite{TorG2020} for efficient generation of the statistics. As in the case of the unraveling of the completely positive master equation, the expected computing time scales as $ O(2d)^2 $ times the number of realizations.

The non anticipating nature of the generated statistics also paves the way to a measurement interpretation. This is possible by associating an instrument \cite{BaLu2005} to the unraveling  e.g. 
directly if the evolution preserves positivity or by embedding in a $ \mathbb{C}^{2}\,\otimes\,\mathcal{H} $ completely positive map as detailed in \cite{DoMG2023} drawing from \cite{BreH2004}. The main limitation of the unraveling is the same that applies to the canonical master equation. The existence result is non constructive. Actual applications of the influence martingale require the knowledge of the canonical  
couplings. In general,  canonical couplings can  be determined via time convolution-less perturbation theory \cite{NeBrWe2021} (see e.g. \cite{SmVa2010} for an application). 

The Gaussian random state vector has the advantage that it is always possible to exactly compute the augmented covariance of the driving Gaussian processes. Genuine quantum effects as interference are handled by introducing a Gaussian process which affects the dynamics only via its complementary covariance. As a consequence, the unraveling is in general not described by an (ostensible) state vector and its adjoint dual but requires to generate the realizations dynamics of the two vector-valued random processes $ \bm{\varphi} $ and $ \bm{\phi} $.  The fact that the unraveling is described by random differential equations driven by time-correlated processes directly reflects the true nature of hidden Markov process (see e.g. \cite{EpMe2002}) of the open system. The time scales involved in the system-environment interaction are not really resolved as it is done in time-convolutionless perturbation theory but rather modeled by the time correlations of significantly lower dimensional Gaussian processes driving the effective random dynamics. The physical interpretation of these correlations is also challenging. They appear to introduce anticipating effects that generically prevent a measurement interpretation \cite{WiGa2008} see however, \cite{MeStLu2020}. 
Finally, the estimation of computing time must take into account the generation of the statistics of time-correlated complex-valued Gaussian processes. This task can be accomplished by means of a set of auxiliary It\^o stochastic differential equations driven by Wiener processes.
Furthermore, {the} computing time depends upon the number of decoherence operators as each of them requires two complex-valued time correlated Gaussian variables.  Algorithms and concentration estimates are not immediately available as for It\^o stochastic differential equations but are actively developed \cite{PrChHuPl2010, StKiKiKe2018}.
\vspace{0.5cm}
\begin{table}[h!]
	\begin{center}
		\label{tab:table1}
\begin{tabular}{|l||c| c||}
	\hline
	 & {\Fontauri Random differential equation} & {\Fontauri   Influence martingale} \\
	\hline\hline
	{\Fontauri Environment}                             & Gaussian                               &  Any  \\
	{\Fontauri Evaluation of couplings}          & Exact from path integral  & Via TCL-PT\\
    {\Fontauri Stochastic process type}         & Hidden Markov                    & Markov \\
    {\Fontauri Numerical methods}                 & Handling of time correlation   & ordinary It\^o SDE \\
   {\Fontauri Measurement interpretation} & No  & Yes \\
	\hline
\end{tabular}
	\caption{Summary of the comparison between the unravelings of section~\ref{sec:comparison}\\
		TCL-PT:  time-convolutionless perturbation theory. \\SDE: stochastic differential equation. }
\end{center}
\end{table}

\section{Conclusions and outlook}

In this paper we reviewed the influence martingale unraveling \cite{DoMG2022,DoMG2023} 
which to the best of our understanding provides the most general way to unravel an open quantum system.
We also derived in a model problem the state-of-the art form \cite{StGr2002,StoJ2004,TilA2017} of the random differential equation \cite{StrW1996,DiSt1997, DiGiSt1998, DiFe2014} which unravels the state operator of a system evolving in a Gaussian environment.
We believe that contrasting the derivations of the unraveling is the most direct avenue to delve into their significance and use.

Unravelings have multifarious applications \cite{WisH1996}.  In our view, two of the most promising are  parameter estimation from measurement records \cite{GaWi2001}, especially in the context of solid state implementation of quantum integrated circuits \cite{WeMuKiScRoSi2016,DoMGPe2019,DoGoMG2020}, and quantum error mitigation \cite{DoLeAnMG2023}. 

Near-term applications of quantum computing with noisy circuits require methods to mitigate or 
cancel errors induced by decoherence or faulty gate operations leading to the loss of quantum resources. 
Ideally, a quantum operation is a pure unitary. We may conceptualize an error model in a given environment as a completely positive map which steers the system state vector away from desired target states. This way of reasoning leads to confront the problem of inverting a non-unitary completely positive map. This is a task that cannot be directly accomplished by means of another completely positive operation.
The trailblazing paper \cite{TeBrGa2016} proposed, among other contributions, that error cancellation can be achieved by subsequent random application of the completely positive maps entering the Wittstock-Paulsen decomposition of the completely bounded map which yields the inverse of the faulty evolution. In this way error mitigation is amenable to classical post-processing without the use of additional quantum resources. The application of the method requires a characterization of the sequence of completely positive operation to be applied and corresponding concentration estimates \cite{JiWaWa2021}. In \cite{DoLeAnMG2023} we implement the idea, also outlined in \cite{DoMG2023}, of accomplishing the same task at the level of unraveling. Clearly, unravelings rely on the knowledge of the possible decoherence channels affecting the system. Thus fidelity robustness with respect to finite accuracy in the knowledge of the decoherence channels has to be taken into account. Our preliminary results encourage us to think that unraveling theory in general, and the influence martingale in particular provide ductile novel tools for quantum error mitigation. Finally, we would like to thank Mi\l osz Michalski for his assistance and patience in editing the final version of the paper.

\section{Acknowledgments}

It is a pleasure to thank Dariusz Chru{\'{s}}ci{\'{n}}ski, Manuel Gessner and Bassano Vacchini with whom PMG had many discussions during the preparation of the manuscript, and Jukka Pekola and the entire PICO group at Aalto University for many discussions and constant encouragement.
We also warmly thank Alberto Barchielli for useful comments on similarities and differences between the influence martingale and his earlier works \cite{BaBe1991,BaPe2010,BaPePe2012} on martingale methods for the unraveling of the completely positive master equation. The authors are also grateful to Erik Aurell and Ingemar Bengtsson for the invitation to contribute to the Lindblad Memorial volume in OSID. 

\appendix

\section*{Appendix}

\setcounter{equation}{0}
\renewcommand{\theequation}{\thesection\arabic{equation}}

\section{Definition of complex-valued Gaussian processes}
\label{ap:cGp}

Let $\bm{\theta}= \big{\{}\bm{\theta}_{t}\big{\}}_{t\,\geq\,0} $ and $ \bm{\vartheta}= \big{\{}\bm{\vartheta}_{t}\big{\}}_{t\,\geq\,0}  $ two $ d $-dimensional zero mean, real-valued correlated Gaussian processes with covariance
\begin{align}
	\mathcal{C}_{t,s}=\begin{bmatrix}
		\operatorname{E}\bm{\theta}_{t}\bm{\theta}_{s}^{\top} &	\operatorname{E}\bm{\theta}_{t}\bm{\vartheta}_{s}^{\top} 
		\\ 
		\operatorname{E}\bm{\vartheta}_{t}\bm{\theta}_{s}^{\top} &	\operatorname{E}\bm{\vartheta}_{t}\bm{\vartheta}_{s}^{\top} 
	\end{bmatrix}
	\nonumber
\end{align}
We may unambiguously define a zero mean complex-valued Gaussian process via the unitary mapping $ \operatorname{T} $
\begin{align}
	\begin{bmatrix}
		\bm{\zeta}_{t}	\\ \bm{\bar{\zeta}}_{t}  
	\end{bmatrix}
	= \operatorname{T}
	\begin{bmatrix}
		\bm{\theta}_{t}	\\  \bm{\vartheta}_{t}
	\end{bmatrix},
	\hspace{1.0cm}
	\operatorname{T}=\frac{1}{\sqrt{2}}\begin{bmatrix}
		\operatorname{1}_{d}	& \imath\,  \operatorname{1}_{d}
		\\
		\operatorname{1}_{d}	& -\,\imath\,  \operatorname{1}_{d}
	\end{bmatrix}
	\nonumber
\end{align}
The transformation yields the covariance of the complex-valued process
\begin{align}
	\mathcal{Z}_{t,s}=
	\begin{bmatrix}
		\operatorname{E}\bm{\zeta}_{t}\bm{\zeta}_{s}^{\dagger} &	\bm{\zeta}_{t}\bm{\zeta}_{s}^{\top}
		\\
		\operatorname{E}\bm{\bar{\zeta}}_{t}\bm{\zeta}_{s}^{\dagger} &	\bm{\bar{\zeta}}_{t}\bm{\zeta}_{s}^{\top}
	\end{bmatrix}
	:=\operatorname{T}\mathcal{C}_{t,s}\operatorname{T}^{\dagger}=\begin{bmatrix}
		\operatorname{Z}_{t,s} & \operatorname{R}_{t,s}
		\\  
		\bar{\operatorname{R}}_{t,s} & \bar{\operatorname{Z}}_{t,s}
	\end{bmatrix}
	\label{cGp:block}
\end{align}
and
\begin{subequations}
	\label{cGp:def}
	\begin{align}
		&\label{cGp:def1}
		\operatorname{Z}_{t,s}=\frac{1}{2}\operatorname{E}\Big{(}\bm{\theta}_{t}\bm{\theta}_{s}^{\top}+	\bm{\vartheta}_{t}\bm{\vartheta}_{s}^{\top}+\imath\,(
		\bm{\theta}_{t}\bm{\vartheta}_{s}^{\top} -\bm{\vartheta}_{t}\bm{\theta}_{s}^{\top})\Big{)}
		\\
		&\label{cGp:def2}
		\operatorname{R}_{t,s}=	\frac{1}{2}\operatorname{E}\Big{(}\bm{\theta}_{t}\bm{\theta}_{s}^{\top}-	\bm{\vartheta}_{t}\bm{\vartheta}_{s}^{\top}+\imath\,(
		\bm{\theta}_{t}\bm{\vartheta}_{s}^{\top}+\bm{\vartheta}_{t}\bm{\theta}_{s}^{\top})\Big{)}
	\end{align}
\end{subequations}
In the information theory literature  \cite{NeMa1993,ScSc2003} {a matrix specified by (\ref{cGp:block}), (\ref{cGp:def})} is called an \textquotedblleft augmented\textquotedblright\ covariance. It is also customary to distinguish between the two following cases.
\begin{itemize}
	\item A Gaussian process is \textquotedblleft proper\textquotedblright\ if
	\begin{align}
		\operatorname{E}(\bm{\theta}_{t}\bm{\theta}_{s}^{\top}-	\bm{\vartheta}_{t}\bm{\vartheta}_{s}^{\top})
		=	\operatorname{E}(\bm{\theta}_{t}\bm{\vartheta}_{s}^{\top}+\bm{\vartheta}_{t}\bm{\theta}_{s}^{\top})=0
		\nonumber
	\end{align}
	In this case, the augmented covariance is fully specified by the \textquotedblleft proper\textquotedblright\ covariance matrix $ \operatorname{Z}_{t,s} $. A proper process comes about when the real Gaussian processes are independent.
	\item A Gaussian process is \textquotedblleft improper\textquotedblright\ if the \textquotedblleft complementary-covariance\textquotedblright\ (or \textquotedblleft pseudo-covariance\textquotedblright\ or \textquotedblleft relation\textquotedblright ) matrix $ \operatorname{R}_{t,s} $
	is non vanishing.
\end{itemize}
In conclusion, a complex matrix with the block structure (\ref{cGp:block}) is an augmented covariance if and only if  it enjoys the following properties 
\begin{enumerate}[label=\textbf{P.\roman*}]
	\item \label{cov:spd}  is positive semi-definite.
	\item \label{cov:diag}  diagonal blocks transform under complex conjugation as required by the definition (\ref{cGp:def1}) of the proper covariance:
	$$ \operatorname{Z}_{t,s}^{\dagger}=\operatorname{Z}_{s,t} $$
	\item \label{cov:non-diag} off-diagonal blocks as symmetric as required by the definition (\ref{cGp:def2}) of the complementary covariance: $$ \operatorname{R}_{t,s}^{\top}=\operatorname{R}_{s,t} $$
\end{enumerate}
The outlined construction  allows us to straightforwardly determine the generating function of the complex-valued process. Given two arbitrary source-fields
$ \bm{a}=\big{\{}\bm{a}_{t}\big{\}}_{t\,\geq\,0} $, $ \bm{b}=\big{\{}\bm{b}_{t}\big{\}}_{t\,\geq\,0} $ we define
\begin{align}
	&\mathcal{E}(\bm{a},\bm{b} )	=\operatorname{E}\exp\left(\int_{0}^{t}\mathrm{d}s\, \bm{a}_{s}^{\top}\bm{\zeta}_{s}+ \bm{b}_{s}^{\top}\bm{\bar{\zeta}}_{s}\right)
	:=
	\nonumber\\
	&		\operatorname{E}\exp\left(\int_{0}^{t}\mathrm{d}s\, \frac{(\bm{a}_{s}+\bm{b}_{s})^{\top}}{\sqrt{2}}\bm{\theta}_{s}+
	\frac{\imath\,(\bm{a}_{s}-\bm{b}_{s})^{\top}}{\sqrt{2}}\bm{\vartheta}_{s} \right)
	\nonumber
\end{align} 
We are then in the position to apply the expression of the generating function of real Gaussian random processes and hence to arrive at
\begin{align}
	\mathcal{E}(\bm{a},\bm{b} )=\exp\left(\frac{1}{2}
	\int_{0}^{t}\mathrm{d}s\int_{0}^{t}\mathrm{d}u\,
		\begin{bmatrix}
		\bm{a}_{s}	\\ \bm{b}_{s}  
	\end{bmatrix}^{\top}
{\operatorname{T}
	\mathcal{C}_{s,u}
	\operatorname{T}^{\top}}
	\begin{bmatrix}
		\bm{a}_{u}	\\ \bm{b}_{u}  
	\end{bmatrix}
	\right)
	\nonumber
\end{align}
which defines the matrix
\begin{align}
	\operatorname{T}\mathcal{C}_{t,s}\operatorname{T}^{\top}
	=\mathcal{Z}_{t,s}\operatorname{T}\operatorname{T}^{\top}
	=
	\begin{bmatrix}
		\operatorname{R}_{t,s} & \operatorname{Z}_{t,s}	
		\\  
		\bar{\operatorname{Z}}_{t,s} & \bar{\operatorname{R}}_{t,s}
	\end{bmatrix}
	\label{cGp:cf}
\end{align}
A common trick in quantum field theory, see e.g. \cite{ZeeA2010}, uses the characteristic function to compute expectation values with respect to the 
Gaussian measure as a formal series of functional derivatives evaluated a zero field
\begin{align}
\operatorname{E}\mathcal{F}(\bm{\zeta},\bm{\bar{\zeta}})=	\mathcal{E}\left (\frac{\delta}{\delta \bm{a}},\frac{\delta}{\delta \bm{b}}\right)\Big{|}_{\bm{a}=\bm{b}=0}\mathcal{F}(\bm{a},\bm{b})\
	\nonumber
\end{align}
Gaussian integration by parts \cite{NuaD1995} can be regarded as formal consequence of this formula
\begin{align}
&	\operatorname{E} \bm{\zeta}_{t}\mathcal{F}(\bm{\zeta},\bm{\bar{\zeta}})=	\mathcal{E}\left (\frac{\delta}{\delta \bm{a}},\frac{\delta}{\delta \bm{b}}\right)\Big{|}_{\bm{a}=\bm{b}=0}\bm{a}_{t} \mathcal{F}(\bm{a},\bm{\bar{b}})
\nonumber\\
&=\int_{0}^{t}\mathrm{d}s\operatorname{E}
	\left( \bar{\operatorname{Z}}_{t,s}\cdot\frac{\delta}{\delta \bm{b}_{s}} + \operatorname{R}_{t,s}\cdot\frac{\delta}{\delta \bm{a}_{s}}\right )\Big{|}_{\bm{a}=\bm{b}=0}\mathcal{F}(\bm{a},\bm{\bar{b}})
	\nonumber
\end{align}

\bibliography{lindblad_osid}{} 
\bibliographystyle{abbrv}

\end{document}